\documentclass[hidelinks,11pt]{article}
\usepackage{fullpage}
\usepackage{amsfonts}     
\usepackage{hyperref}
\hypersetup{
	colorlinks,
	linkcolor={black},
	citecolor={blue!50!black}
}
\usepackage{amsmath, amssymb, graphicx, url} 
\usepackage{amsthm} 
\usepackage{nicefrac}       
\usepackage{microtype}     
\usepackage{color}
\usepackage{algorithm}
\usepackage{algpseudocode}
\usepackage{mathtools}
\usepackage{tikz}
\usepackage{pgfplots}
\usepackage{cite}
\usepackage{epstopdf}
\usepackage{bm}
\usepackage{tabulary}
\usepackage[mathscr]{euscript}

\newtheorem{theorem}{\bf Theorem}
\newtheorem{lemma}{\bf Lemma}
\newtheorem{proposition}{\bf Proposition}
\newtheorem{corollary}{\bf Corollary}
\newtheorem{remark}{\bf Remark}
\newtheorem{problem}{\bf Question}
\newtheorem{assumption}{\bf Assumption}
\newtheorem{defin}{\bf Definition}

\newcommand{\paren}[1]{\ensuremath{\left( #1\right)}}

\newcommand{\clint}[1]{\ensuremath{\left[ #1\right]}}
\newcommand{\set}[1]{\ensuremath{\left\{ #1\right\}}}
\newcommand{\matr}[1]{\ensuremath{\clint{\begin{array} #1 \end{array}}}}
\newcommand{\norm}[1]{\ensuremath{\left\| #1\right\|}}
\newcommand{\snorm}[1]{\ensuremath{\| #1\|}}
\newcommand{\abs}[1]{\ensuremath{\left| #1\right|}}

\newcommand{\poly}{\ensuremath{\mathrm{poly}}}
\newcommand{\KL}{\ensuremath{\mathrm{KL}}}
\newcommand{\rank}{\ensuremath{\mathrm{rank}}}
\newcommand{\N}{\ensuremath{\mathcal{N}}}
\newcommand{\R}{\ensuremath{\mathbb{R}}}

\renewcommand{\P}{\ensuremath{\mathbb{P}}}

\newcommand{\F}{\ensuremath{\mathcal{F}}}
\newcommand{\E}{\ensuremath{\mathbb{E}}}

\newcommand{\C}{\ensuremath{\mathcal{C}}}
\newcommand{\CC}{\ensuremath{\mathscr{C}}}

\newcommand{\T}{\ensuremath{\mathcal{T}}}

\graphicspath{{./figures/}}
\title{\bf  Linear Systems can be Hard to Learn}
\usepackage{times}
\author{Anastasios~Tsiamis and George~J.~Pappas 
	\thanks{The authors  are   with   the   Department   of   Electrical   and   Systems  Engineering,  University  of  Pennsylvania,  Philadelphia,  PA  19104.
		Emails: \{atsiamis,pappasg\}@seas.upenn.edu. This work is supported by the NSF-Simons grant 2031985 and the AFOSR Assured Autonomy grant.}
		}

\begin{document}

\maketitle

\begin{abstract}
In this paper, we investigate when system identification is statistically easy or hard, in the finite sample regime. Statistically easy to learn linear system classes have sample complexity that is polynomial with the system dimension. Most prior research in the finite sample regime falls in this category, focusing on systems that are directly excited by process noise.  Statistically hard to learn linear system classes have worst-case sample complexity that is at least exponential with the system dimension, regardless of the identification algorithm.  Using tools from minimax theory, we show that classes of linear systems can be hard to learn. Such classes include, for example, under-actuated or under-excited systems with weak coupling among the states. Having classified some systems as easy or hard to learn, a natural question arises as to what system properties  fundamentally affect the hardness of system identifiability.  Towards this  direction, we characterize how the controllability index of linear systems affects the  sample  complexity of  identification.  More specifically, we show that the sample complexity of robustly controllable linear systems is upper bounded by an exponential function of the controllability index.  This implies that identification is easy for classes of linear systems with small controllability index and potentially hard if the controllability index is large. Our analysis is based on recent statistical tools for finite sample analysis of system identification as well as a novel lower bound that relates controllability index with the least singular value of the controllability Gramian. 
\end{abstract}

\section{Introduction}\label{sec:intro}

Linear system identification focuses on using input-output data samples for learning dynamical systems of form:
 \begin{equation}\label{eq:system}
\begin{aligned}
    x_{k+1}&=Ax_{k}+Bu_{k}+Hw_k,
\end{aligned}
\end{equation} 
where $x_k$ represents the state, $u_k$ represents the control signal, and $w_k$ is the process noise.  The statistical analysis of system identification algorithms has a long history~\cite{Ljung1999system}. Until recently, the main focus was providing guarantees for the convergence of system identification in the \emph{asymptotic regime}~\cite{deistler1995consistency,bauer1999consistency,chiuso2004asymptotic},  when the number of collected samples $N$ tends to infinity. Under sufficient persistency of excitation~\cite{bai1985persistency}, system identification algorithms converge and the asymptotic bounds capture very well how the identification error decays with $N$ qualitatively. 

However, our standard asymptotic tools (e.g. the Central Limit Theorem), do not always capture all finite-sample phenomena~\cite[Ch 2]{vershynin2018high}. Moreover, the identification error depends on various system theoretic constants, like the state space dimension $n$, which might be hidden under the big-$O$ notation in the asymptotic bounds.
As a result, system identification limitations, like the curse of dimensionality, although known to practitioners, are not always reflected in the theoretical asymptotic bounds.

With the advances in high-dimensional statistics~\cite{vershynin2018high}, there has been a recent shift from asymptotic analysis with infinite data to statistical analysis of system identification with finite samples.
Over the past two years there have been significant advances in understanding finite sample system identification for both fully-observed systems~\cite{campi2002finite,dean2017sample,simchowitz2018learning,faradonbeh2018finite,sarkar2018fast,fattahi2019learning,jedra2019sample,wagenmaker2020active} as well as partially-observed systems~\cite{oymak2018non,sarkar2019finite,simchowitz2019semi,tsiamis2019finite,lee2019non,zheng2020non,lee2020improved,lale2020logarithmic,kozdoba2019line,tsiamis2020online}. A tutorial can be found in~\cite{matni2019tutorial}. The above approaches offer mainly \emph{data-independent} bounds which reveal how the state dimension $n$ and other system theoretic parameters affect the sample complexity of system identification \emph{qualitatively}. This is different from finite sample data-dependent bounds-see for example bootstrapping~\cite{dean2017sample} or~\cite{care2018finite}, which might be more tight and more suitable for applications but do not necessarily reveal this dependence.

Despite these advances, we still do not fully understand the fundamental limits of when identification  is  easy  or  hard.  In this paper, we define as statistically easy, classes of systems whose finite-sample  complexity is  polynomial  with  the  system  dimension.  Most  prior research in the finite-sample analysis of fully observed systems falls  in  this  category by assuming system~(\ref{eq:system}) is fully excited by the process noise $w_k$. We define as statistically hard, classes of linear systems whose worst-case sample complexity is at least exponential  with  the  system  dimension,  regardless  of the learning algorithm. Using recent tools from  minimax  theory~\cite{jedra2019sample}, we show that classes of linear systems which are statistically hard to learn do indeed exist.  Such system classes include, for  example,  under-actuated systems with weak state coupling. The fact that linear systems  may contain exponentially hard classes has implications for broader classes of systems, such as nonlinear systems, as well as control algorithms, such as the linear quadratic regulator~\cite{recht2018tour} and reinforcement learning~\cite{du2019good,jiang2017contextual}. 

By examining classes of linear systems that are statistically easy or hard, we quickly arrive at the conclusion that system theoretic properties, such as controllability, fundamentally  affect the hardness of identification.
In fact, as we show in the paper, structural properties like the controllability index can crucially affect learnability, determining whether a problem is hard or not.
In summary, our contributions are the following:

\noindent \textbf{--Learnability of dynamical systems.} We define two novel notions of learnability for classes of dynamical systems. A class of systems is easy to learn if it exhibits polynomial sample complexity with respect the state dimension $n$. It is hard to learn if for any possible learning algorithm it has exponential worst-case complexity.

\noindent\textbf{--Exponential sample complexity is possible.}
We identify classes of under-actuated linear systems whose worst-case sample complexity increases exponentially with the state dimension $n$ regardless of learning algorithm.  These hardness results hold even for robustly controllable systems.

\noindent\textbf{--Controllability index affects sample complexity.}
 We prove that under the least squares algorithm, the sample complexity is upper-bounded by an exponential function of the system's controllability index. This implies that if the controllability  index is small~$O(1)$ (with respect to the dimension $n$), the  sample complexity is guaranteed to be polynomial generalizing previous cases.  If, however, the index grows linearly $\Omega(n)$, then there exist non-trivial linear systems which are exponentially hard to identify. 

\noindent\textbf{--New controllability Gramian bound}
Our sample complexity upper bound is a consequence of a new result that is of independent, system theoretic interest. We prove that for robustly controllable systems, the least singular value of the controllability Gramian can grow at most exponentially with the controllability index.  Although it has been observed empirically that the Gramian might be affected by the curse of dimensionality~\cite{baggio2019data}, to the best of our knowledge this theoretical bound is new and has implications beyond system identification.

\textit{\textbf{Notation:}}
The transpose operation is denoted by $(\cdot)'$ and the complex conjugate by $*$. By $e_i\in\R^{n}$ we denote the $i-$th canonical vector. By $\sigma_{\min}$ we denote the least singular value. $\succeq$ denotes comparison in the positive semidefinite cone. The identity matrix of dimension $n$ is denoted by $I_n$. The spectral norm of a matrix $A$ is denoted by $\snorm{A}_2$. The notion of controllability and other related concepts are reviewed in the Appendix.

\section{Learnability of System Classes}\label{sec:formulation}
Consider system~\eqref{eq:system}, where $x_k\in \R^n$ is the state and $u_k\in\R^p$ is the input. By $w_k\in\R^{r}$ we denote the process noise which is assumed to be Gaussian, i.i.d. with covariance $I_r$.  Without loss of generality the initial state is assumed to be zero $x_0=0$.
\begin{assumption}\label{ass_boundedness}
All state parameters are bounded: $\snorm{A}_2,\snorm{B}_2,\snorm{H}_2\le M$, for some positive constant $M>0$. The noise has unknown dimension $r$ and can be degenerate $r\le n$. All parameters $A,B,H,r$ are considered unknown. Matrices $B,H$ have full column rank $\rank(B)=p\le n$, $\rank(H)=r\le n$. We also assume that the system is non-explosive $\rho(A)\le 1$. Finally, we assume that the control inputs have bounded energy $\E u'_tu_t\le M$.
\end{assumption}
This setting is rich enough to provide insights about the difficulty of the general learning problem. To simplify the setting we assume that the system is non-explosive. The analysis of unstable systems is left for future research.

A system identification (SI) algorithm $\mathcal{A}$ receives a finite number $N$ of input-state data $(x_{0},u_{0}),\dots,(x_N,u_N)$ generated by system~\eqref{eq:system}, and returns an estimate of the unknown system's parameters $\hat{A}_N,\hat{B}_N,\hat{H}_N$. We denote by $N$ the number of collected input-state samples, which are generated during a single roll-out of the system, that is a single trajectory of length $N$.
For simplicity, we focus only on the estimation of matrix $A$ in this paper.

Our goal is to study when the problem of system identification is fundamentally easy or hard.
The difficulty is captured by the sample complexity, i.e. how many data $N$ do we need to achieve small identification error with high probability.  
Formally, let $\epsilon>0$, $0<\delta<1$ be the accuracy and confidence parameters respectively. Then, the sample complexity is the smallest possible number of samples $N$ such that with probability at least $1-\delta$ we can estimate $A$ with small error $\snorm{A-\hat{A}_N}\le \epsilon$.  
Naturally, the sample complexity increases as the accuracy/confidence parameters $\epsilon,\delta$ decrease. The sample complexity also increases in general with the state-space dimension $n$ and the bound $M$ on the state space parameters.

Ideally, the sample complexity should grow slowly with $n,M,\epsilon^{-1},\delta^{-1}$. 
Inspired by Provably Approximately Correct (PAC) learning~\cite{shalev2014understanding,dann2017unifying}, we classify an identification problem as easy when the sample complexity depends polynomially on $n,M,\epsilon^{-1},\delta^{-1}$. 
For brevity we will use the symbol $S$ to denote the tuple $S=(A,B,H)$. Let $\P_{S}$ denote the probability distribution of the input-state data when the true parameters of the system are equal to $S$ and we apply a control law $u_t\in\mathcal{F}_t$, where $\F_t\triangleq \sigma(x_0,u_0,\dots,u_{t-1},x_t)$ is the sigma algebra generated by the previous outputs and inputs. By $\CC_n$ we will denote a class of systems with dimension $n$. 
\begin{defin}[$\poly$-learnable classes]\label{def:poly}
	Let $\CC_n$ be a class of systems. Consider a trajectory of input-state data $(x_{0},u_{0}),\dots$,$(x_N,u_N)$, which are generated by a system $S$ in $\CC_n$ under some control law $u_t\in\mathcal{F}_t$, $t\le N$. 
	We call the class $\CC_n$ $\poly(n)-$learnable if there exists an identification algorithm such that the sample complexity is polynomial: for any confidence $0\le \delta<1$ and any tolerance $\epsilon>0$:
	\begin{align}
		&\sup_{S\in\CC_n}\P_{S}(\snorm{A-\hat{A}_N}\ge \epsilon)\le \delta\label{eq:objective},\\
		&\text{ for }N\ge \mathrm{poly}(n,1/\epsilon,\log 1/\delta, M)\nonumber,
	\end{align}
	where $\poly(\cdot)$ is some polynomial function.
\end{defin}
Definition~\ref{def:poly} provides an intuitive definition for a class $\CC_n$ of linear systems whose system identification problem is easy. To prove that a class of systems $\CC_n$ is easy, it suffices to provide one algorithm that performs well for any system $S \in\CC_n$ in the sense that it requires \textbf{at most} a polynomial number of samples.  This means that we should obtain sample complexity \textbf{upper bounds} across all $S \in\CC_n$ which is what the the supremum over $S\in \CC_n$ achieves in~\eqref{eq:objective}.  Otherwise, we can construct trivial algorithms that perform well only on one system and fail to identify the other.

\begin{figure}[t] \centering
	\definecolor{mycolor1}{rgb}{0.00000,0.44700,0.74100}%
\definecolor{mycolor2}{rgb}{0.85000,0.32500,0.09800}%
\definecolor{mycolor3}{rgb}{0.92900,0.69400,0.12500}%
\resizebox{0.55\columnwidth}{!}{\begin{tikzpicture}
\begin{axis}[%
		width=6.028in,
		height=2.754in,
		at={(1.011in,0.642in)},
		scale only axis,
		xmin=5,
		xmax=12,
		tick label style={font=\LARGE},
			ylabel style={font=\huge},
	xlabel style={font=\huge},
		ymode=log,
		ymin=10,
		ymax=10000,
		xtick={5,6,7,8,9,10,11,12},
		yminorticks=true,
		xlabel={dimension $n$},
		ylabel={samples $N$},
		xmajorgrids,
		ymajorgrids,
		axis background/.style={fill=white},
		legend style={legend cell align=left, align=left, draw=white!15!black,font=\huge,legend pos=north west}
		]
		\addplot [color=mycolor1,line width=1.5pt]
		table[row sep=crcr]{%
5	11\\
6	16\\
7	24\\
8	45\\
9	125\\
10	415\\
11	1573\\
12	5499\\
		};
		\addlegendentry{$\epsilon$=0.1}
			\addplot [color=mycolor2, dashed,line width=2pt]
		table[row sep=crcr]{%
5	10\\
6	14\\
7	19\\
8	31\\
9	68\\
10	198\\
11	692\\
12	2641\\
		};
		\addlegendentry{$\epsilon$=0.15}
			\addplot [color=mycolor3,dashdotted,line width=2pt]
		table[row sep=crcr]{%
5	9\\
6	12\\
7	17\\
8	26\\
9	47\\
10	124\\
11	403\\
12	1405\\
		};
		\addlegendentry{$\epsilon$=0.2}
	\end{axis}
	\end{tikzpicture}}
 	\caption{The minimum number of samples $N$ such that the (empirical) average error $\mathbb{E}\snorm{A-\hat{A}_N}_2$, for identifying~\eqref{eq:anisotropic}, is less than $\epsilon$. The sample complexity appears to be increasing exponentially with the dimension $n$ under the least squares algorithm.}
	\label{Fig:motivational_example}
\end{figure}
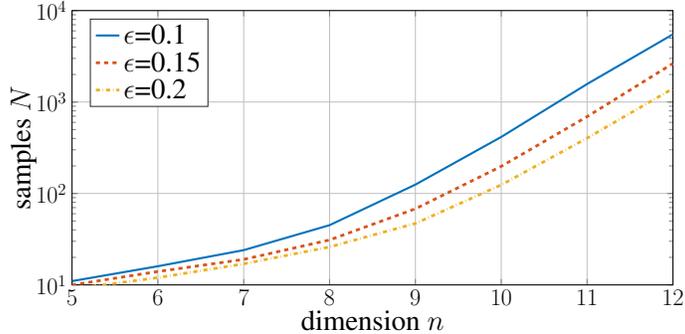

In recent work~\cite{simchowitz2018learning,sarkar2018fast,fattahi2019learning}, it was shown that under the least squares algorithm, the sample complexity of learning linear systems is polynomial.  As we review in Section~III, these results  hold for classes of linear systems 
where the noise is isotropic and hence directly exciting all states.

However, if we relax the last assumption it turns out that the sample complexity might degrade dramatically. To raise this issue, consider the following example. Let $J_n(1)$ be a Jordan block of size $n$ with eigenvalue $1$ and let $e_n$ be the $n-$th canonical vector. We simulate the performance of least squares identification for the system
\begin{equation}
    x_{k+1}=0.5 J_n(1)x_{k}+e_n(u_k+w_k) \label{eq:anisotropic}
\end{equation}
Note that in system (\ref{eq:anisotropic}) the process noise is no longer isotropic. Figure~\ref{Fig:motivational_example} shows the minimum number of samples $N$ required to achieve (empirical) average error $\mathbb{E}\snorm{A-\hat{A}_N}\le \epsilon$ (the details of the simulation can be found in Section~\ref{sec:simulations}). It seems that the sample complexity increases exponentially rather than polynomially. Are the results in Figure~\ref{Fig:motivational_example} due to the choice of the algorithm or is there a fundamental limitation for all system identification algorithms?  We pose the following fundamental problem.

\begin{problem}\label{problem}
	Do there exist classes of linear systems which are hard to learn, meaning not $\poly$-learnable by any system identification algorithm? Furthermore, can the sample complexity for a class of linear systems be exponential with state dimension $n$?
\end{problem}

A class of linear systems $\CC_n$ that is not $\poly$-learnable will be viewed as hard. By negating Definition $\ref{def:poly}$, this notion of hardness means that given {\em any} system identification algorithm, there exist instances $S\in \CC_n$ that cannot have polynomial sample complexity. In other words, a system class $\CC_n$ is classified as hard when its impossible to find any system identification algorithm that achieve polynomial sample complexity for all $S\in \CC_n$. This can be viewed as a fundamental statistical limitation for the chosen class of systems $\CC_n$. 

Motivated by Figure~\ref{Fig:motivational_example}, we define an important subclass of hard problems, namely linear system classes that have worst-case sample complexity that grows exponentially with the dimension $n$ regardless of identification algorithm choice.

\begin{defin}[$\exp$-hard classes]\label{def:exp}
	Let $\CC_n$ be a class of systems of dimension $n$. Consider a trajectory of input-output data $(x_{0},u_{0}),\dots$,$(x_N,u_N)$, which are generated by a system $S$ in $\CC_n$ under some control law $u_t\in\mathcal{F}_t$, $t\le N$. 
	We call a class $\CC_n$ of systems $\exp(n)$-hard if the sample complexity is at least exponential with the dimension $n$: there exist confidence $0\le \delta<1$ and tolerance $\epsilon$ parameters such that for any identification algorithm:
	\begin{align*}
		&\sup_{S\in\CC_n}\P_{S}(\snorm{A-\hat{A}_N}\ge \epsilon)\le \delta,\\
		&\text{ only if }N\ge \mathrm{exp}(n),
	\end{align*}
	where $\mathrm{exp}(n)$ denotes an exponential function of $n$.
\end{defin}
System classes $\CC_n$ that are $\exp$-hard are an important subset of hard system classes as they are clearly not $\poly$-learnable.  However, not all classes that are not $\poly$-learnable are $\exp$-hard.

In order to show that a class of systems $\CC_n$ is $\exp$-hard, one must show that for \textbf{any} system identification algorithm the worst-case sample complexity is at least exponential in state dimension $n$. 
Contrary to $\poly$-learnable problems,
for exponential hardness we should establish sample complexity \textbf{lower bounds}.

In this paper, we first address Question~\ref{problem} and show that $\exp$-hard classes of linear systems do indeed exist. While this can be viewed as a fundamental statistical limitation for all system identification algorithms, our results open a new direction of research that classifies when linear systems are easy to learn and when they are hard to learn.  This leads to the following important question addressing in this paper:.
 \begin{problem}\label{problem:avoid}
When is a class of linear systems $\CC_n$ guaranteed to be $\poly$-learnable?
 \end{problem}
Based on prior work, we already have partial answers to Question~\ref{problem:avoid} as we know that linear systems with isotropic noise are $\poly$-learnable. In Section~\ref{sec:guaranteed}, we seek to broaden the classes of $\poly$-learnable systems and discover their relation to fundamental system theoretic properties such as controllability.

While Definitions~\ref{def:poly},~\ref{def:exp} are inspired by PAC learning, they have a different flavor. One of the differences is that the guarantees in Definitions~\ref{def:poly},~\ref{def:exp} are stated in terms of recovering the state-space parameters, while in PAC learning, they would be stated in terms of the prediction error of the learned model or informally $\sum_{k=0}^{N-1}\E\snorm{x_k-\hat{A}x_{k-1}-\hat{B}u_{k-1}}^2$. 

\section{Directly-excited systems are poly-learnable}\label{sec:isotropic}

In this section, we revisit state-of-the-art results in finite-sample complexity for fully-observed linear systems and re-establish that they all lead to polynomial sample complexity. In prior work~\cite{simchowitz2018learning,sarkar2018fast,fattahi2019learning}, the class of linear systems considered assumes that the stochastic process noise is isotropic, i.e. $HH'=\sigma^2_w I_n$. Since all states are directly excited by the process noise, all modes of the system are captured sufficiently in the data. To obtain polynomial complexity, it suffices to use the least squares identification algorithm
\begin{equation}\label{eq:least_squares}
\matr{{cc}\hat{A}_N&\hat{B}_N}=\arg\min_{\set{F,G}} \sum_{t=0}^{N-1}\snorm{x_{t+1}-Fx_t-Gu_t}^2_2
\end{equation}
with white noise inputs $u_t\sim\N(0,\sigma^2_u I)$. 
Based on the algorithm analysis from~\cite{simchowitz2018learning}, let $k$ be a fixed time index which is much smaller than the horizon $N$ (see Theorem 2.1 in~\cite{simchowitz2018learning} for details). Let $0<\delta<1$ and $\epsilon$ be the confidence and accuracy parameters respectively. Then, with probability at least $1-\delta$, the error is $\snorm{A-\hat{A}_N}_2\le \epsilon$ if:
\[
N\ge \frac{c\sigma^2_w}{\sigma_{\min}(\Gamma_k)}\frac{1}{\epsilon^2}\paren{n\log\frac{n}{\delta}+\log\det(\Gamma_{N}\Gamma^{-1}_k)},
\]
where $c$ is a universal constant, and  $\Gamma_k=\sigma^2_u\Gamma_k(A,B)+\sigma^2_w\Gamma_k(A,I_n)$ is the (combined) controllability Gramian.  Uunder the isotropic noise assumption, the least singular value of the Gramian $\Gamma_k$ is bounded away from zero, $\sigma_{\min}(\Gamma_k)\ge \sigma^2_{w}$. 

In a slight departure from~\cite{simchowitz2018learning,sarkar2018fast,fattahi2019learning}, we can  show that the determinant of the Gramian $\det(\Gamma_N)$ can only increase at most polynomially with the number of samples $N$ and exponentially with state dimension $n$.  This is a direct consequence of the following lemma, which is a new result.
\begin{lemma}\label{lem:powers_of_A}
 Let $A\in\R^{n\times n}$ have all eigenvalues inside or on the unit circle, with $\norm{A}_2\le M$. Then, the powers of matrix $A$ are bounded by:
\begin{equation}\label{eq:powers_bound}
    \norm{A^k}_2\le (ek)^{n-1}\max\set{M^{n},1}
\end{equation}   
\end{lemma}
Lemma~\ref{lem:powers_of_A} enables us to eliminate the dependence on the condition number of the Jordan form's similarity transformation, which exists in prior bounds and can be arbitrarily large. We avoid this dependence by using the Schur form of $A$~\cite{horn2012matrix}. While this does not alter the already known sample complexity results, it allows us to have sample complexity bounds that are uniform across all systems that satisfy Assumption~\ref{ass_boundedness}.

As a result of Lemma~\ref{lem:powers_of_A}, we obtain that the system identification problem for linear systems with isotropic noise has polynomial sample complexity. The result can be broadened to the more general case of direct excitation, where the covariance is lower bounded by $HH'+BB'\succeq \sigma_w^2I_n$, for some $\sigma_w>0$, as the following theorem states. 
\begin{theorem}[Directly-excited]\label{thm:directly_excited}
Consider the class $\CC_n$ of directly-excited systems $S=(A,B,H)\in\R^{n\times(n+p+r)}$ such that Assumption~\ref{ass_boundedness} is satisfied with covariance $HH'+BB'\succeq\sigma^2_w I_n$, for some $\sigma_w> 0$. The class $\CC_n$ is $\poly-$learnable under the least squares system identification algorithm with white noise input signals $u_k\sim\mathcal{N}(0,I_p)$.
\end{theorem}
\begin{proof}
It follows as a special case of Theorem~\ref{thm:upper_exponential} for controllability index $\kappa=1$.
\end{proof}
Directly excited systems includes fully-actuated systems (number of inputs equal to the number of states $p=n$), or systems with isotropic noise as special cases.
However, having direct excitation might not always be the case. The combined noise and input matrices might be rank-deficient. For example, we might have actuation noise as in:
\[
x_{t+1}=Ax_t+B(u_t+w_t).
\]
In general, the noise might be ill-conditioned (zero across certain directions), while it might be physically impossible to actuate every state of the system. We call such systems underactuated or under-excited. It might still be possible to identify underactuated systems, e.g. if the pair $(A,\matr{{cc}H&B})$ is controllable. However, as we prove in the next section, the identification difficulty might increase dramatically.
\section{Exp-hard system classes}\label{sec:hard}
In this section, we show that there exist common classes of linear systems which are impossible or hard to identify with a finite amount of samples. As we will see, this can happen  when systems are under-actuated and under-excited. When only a limited number of system states is directly driven by inputs (or excited by noise) and the remaining states are only indirectly excited, 
then identification can be inhibited. 
 \subsection{Controllable systems with infinite sample complexity}

For presentation simplicity, let us assume that there are no exogenous inputs $B=0$. Similar results also hold when $B\neq 0$--see Remark~\ref{rem:inputs}. 
To fully identify the unknown matrix $A$, it is necessary that the pair $(A,H)$ is controllable. Furthermore, let's assume that the noise is meaningful, that is $\sigma_{\min}(H)\ge \sigma$ for some $\sigma>0$.
However, controllability of $(A,H)$ and $\sigma_{\min}(H)\ge \sigma$ are not sufficient to ensure system identification from a finite numer of samples.  
The following, perhaps unsurprising theorem, shows that for this class of linear systems, the worst-case sample complexity is infinite. 
 
 \begin{theorem}[Controllability is not sufficient for finite sample complexity]\label{thm:trivial_example}
Consider the class $\CC_n$ of systems $S=(A,H)\in\R^{n\times(n+r)}$ such that Assumption~\ref{ass_boundedness} is satisfied with $(A,H)$ controllable, and $\sigma_{\min}(H)\ge \sigma$ for some $\sigma>0$. For any system identification algorithm the sample complexity is infinite: there exist a failure probability $0\le \delta<1$ and a tolerance $\epsilon>0$ such that we cannot achieve
\begin{align}
		&\sup_{S\in\CC_n}\P_{S}(\snorm{A-\hat{A}_N}\ge \epsilon)\le \delta \nonumber
	\end{align}
with a finite number of samples $N$.
\end{theorem}
Theorem~\ref{thm:trivial_example} clearly shows that we may need stronger notions of controllability, as done in Section~\ref{sec:exp-classes}, in order to find classes of systems whose sample complexity is finite. 
The proof of Theorem~\ref{thm:trivial_example} uses tools from minimax theory~\cite{jedra2019sample}. Adapting these tools in our setting results in the following. 
 \begin{lemma}[Minimax bounds]\label{lem:minimax}
Let $\CC_n$ be a class of systems. Consider a confidence $0<\delta<1$ and an accuracy parameter $\epsilon>0$. Denote by $S_1,S_2\in \CC_n$ any pair of two systems with $A_1,H_1$, $A_2,H_2$ the respective unknown matrices, such that $\snorm{A_1-A_2}\ge 2\epsilon$. 
 Let $\KL(\P_{S_1},\P_{S_2})$ be the Kullback-Leibler divergence between the probability distributions of the data when generated under $S_1,S_2$ respectively. Then for any identification
 algorithm
 \begin{align}
		&\sup_{S\in\CC_n}\P_{S}(\snorm{A-\hat{A}_N}\ge \epsilon)\le \delta \nonumber
	\end{align}
holds only if
 \begin{align}
 \label{eq:pairwise_lower_bound}
 &\KL(\P_{S_1},\P_{S_2})\ge \log \frac{1}{3\delta},
 \end{align}
for all such pairs $S_1,S_2\in\CC_n$.  
 \end{lemma}
 \begin{proof}
Let $S_1,S_2$ be any pair satisfying the conditions. We trivially have that:
\[
\sup_{S\in\CC_n}\P_{S}(\snorm{A-\hat{A}_N}\ge \epsilon)\le \delta
\]
only if
\[
\sup_{S\in\set{S_1,S_2}}\P_{S}(\snorm{A-\hat{A}_N}\ge \epsilon)\le \delta.
\]
The remaining proof is identical to~\cite[Proposition 2]{jedra2019sample}, where we replaced constant $2.4$ with $3$ for simplicity and we did not expand the expression for $\KL(\P_{S_1},\P_{S_2})$ explicitly (term $\E_A(L_t)$ in~\cite{jedra2019sample}).
\end{proof}
Intuitively, to find difficult learning instances we construct systems which are sufficiently separated ($2\epsilon$ away). Meanwhile, the systems should be similar enough to generate data with as indistinguishable distributions as possible (small KL divergence). 
If the system is hard to excite, then the distributions of the states will look similar under many different matrices $A$, leading to smaller KL-divergence.
Unless we bound the pair $(A,H)$ away from uncontrollability, it might be impossible to satisfy~\eqref{eq:pairwise_lower_bound} for all pairs of systems with a finite number of samples. For example consider:	\[
  A=\matr{{ccc}0&\alpha&0\\0&0&\beta\\0&0&0},\,H=\matr{{cc}1&0\\0&0\\0&1},
  \] 
 It requires an arbitrarily large number of samples to learn $\alpha$
 if the coupling $\beta$ between $x_{t,2}$ and $x_{t,3}$ is arbitrarily small. The distribution of $x_{t,1}$ remains virtually the same as we perturb $\alpha$, since the state $x_{t,2}$ is under-excited for small $\beta$.
 
 \subsection{Robustly controllable systems can be exp-hard}
 \label{sec:exp-classes}
 
 Theorem~\ref{thm:trivial_example} implies that we need to bound the system away from uncontrollability in order to obtain non-trivial sample complexity bounds.  In order to formulate this, we review the notion of distance from uncontrollability, which is the norm of the smallest perturbation that makes $(A,H)$ uncontrollable.
  \begin{defin}[Distance from uncontrollability~\cite{eising1984metric}]
Let $(A,H)\in\R^{n\times(n+r)}$ be controllable. Then, the distance from uncontrollability is given by:
\begin{equation}
\begin{aligned}
    d(A,H)&\triangleq \inf \left\{\snorm{\matr{{cc}\Delta A&\Delta H}}_2:\right.\\
   & \left.(A+\Delta A,H+\Delta H)\text{ uncontrollable}\right\},
\end{aligned}
\end{equation}
where perturbations $(\Delta A,\Delta H)\in\mathbb{C}^{n\times(n+r)}$ are complex.
\end{defin}
Let us now consider linear systems that are robustly controllable. That is, classes of controllable linear systems whose distance from uncontrollability is lower bounded.  The lower bound is allowed to degrade gracefully (polynomially) with the system  dimension $n$. 
\begin{assumption}[Robust Controllability]\label{ass:away_from_uncontrollability}
Assume that system $(A,H)$ is robustly controllable, that is $(A,H)\in\R^{n\times(n+m)}$ is $\mu$-away from uncontrollability:
\begin{equation}
    d(A,H)\ge \mu,
\end{equation}
for some positive $\mu\ge 0$, with $\mu^{-1}\le \poly(n)$.
\end{assumption}
Assumption~\ref{ass:away_from_uncontrollability} is not restrictive as long as we allow the bound to degrade with the dimension. Common systems like the $n-$th order integrator have distance that degrades linearly with $n$--see Lemmas~\ref{lem:integrator_distance_to_uncontrollability},~\ref{lem:integrator_like_distance_to_uncontrollability} in the Appendix.
However, even for system classes that satisfy Assumption~\ref{ass:away_from_uncontrollability}, the next theorem shows that system identification can be $\exp$-hard. 
   \begin{theorem}[Exp(n)-hard classes]\label{thm:lower_exponential}
Consider the set $\CC_n$ of systems $S=(A,H)$ such that Assumptions~\ref{ass_boundedness},~\ref{ass:away_from_uncontrollability} are satisfied with $d(A,H)\ge\mu=8(n+1)^{-1}$. Then, for any system identification algorithm $\mathcal{A}$ the sample complexity is exponential in the state dimension $n$. There exist a confidence $0\le \delta<1$ and a tolerance $\epsilon>0$ such that
\begin{align}
		&\sup_{S\in\CC_n}\P_{S}(\snorm{A-\hat{A}_N}\ge \epsilon)\le \delta \nonumber
	\end{align}
is satisfied only if
\[
N\ge    \frac{4^{n-3}}{3\epsilon^2}\log \frac{1}{\delta}.
\]
\end{theorem}
Theorem~\ref{thm:lower_exponential} shows that even for robustly controllable classes of linear systems satisfying Assumptions~\ref{ass_boundedness},~\ref{ass:away_from_uncontrollability}, any system identification algorithm will have worst-case sample complexity that depends exponentially on the system dimension $n$.  The proof of Theorem~\ref{thm:lower_exponential} is based once more on minimax theory used in Lemma~\ref{lem:minimax}.

The reason for this learning difficulty is due to the need for indirect excitation. Consider, for example, chained systems, where every state indirectly excites the next one. If the states are weakly-coupled, then the exploratory signal (noise or input) attenuates
exponentially fast along the chain. As a concrete example, consider the following system for $\rho<0.5$:
 \begin{equation}\label{eq:difficult_system}
	A=\matr{{cccccc}\rho &\rho&0&\cdots&0&0\\0& \rho&\rho&\cdots&0&0\\& & &\ddots&\\0&0&0&\cdots&\rho&\rho\\0&0&0&\cdots&0&\rho},\, H=\matr{{cc}1&0\\\vdots&\vdots\\0&\rho} 
\end{equation}
 which satisfies Assumptions~\ref{ass_boundedness},~\ref{ass:away_from_uncontrollability}.
Matrix $A$ has a chained structure with weak coupling between the states. Noise can only excite states $x_{t,1},x_{t,n}$ directly. Until the exploratory noise signal reaches $x_{t,2}$ it decreases exponentially fast with the dimension $n$. As a result, it is difficult to learn $A_{12}$ due to lack of excitation. In terms of Lemma~\ref{lem:minimax}, the distribution of $x_{t,1}$ will remain virtually the same if we perturb $A_{12}$ since $x_{t,2}$ is under-excited.

 \begin{remark}[Exogenous inputs]\label{rem:inputs}
 When $B\neq 0$ similar results hold but with an additional interpretation. Consider system~\eqref{eq:difficult_system} but with $H=e_1$, $B=
 \rho e_n$. Then, if we apply white-noise input signals we have two possibilities: i) the control inputs have bounded energy per Assumption~\ref{ass_boundedness} but we suffer from exponential sample complexity or ii) we obtain polynomial sample complexity but we allow the energy of the inputs to increase exponentially with the dimension. 
 From this alternative viewpoint a system is hard to learn if it requires exponentially large control inputs.
 
 \begin{remark} The constant $8$ in $8(n+1)^{-1}$ in the statement of Theorem~\ref{thm:lower_exponential} is not important in our analysis. We could modify Theorem~\ref{thm:lower_exponential} so that $8$  can be replaced by any smaller constant. In particular, we can decrease $8$ by considering systems with smaller chains, which still have exponential sample complexity. Instead of system~\eqref{eq:difficult_system}, we can consider for example the following. Let $J_{\lfloor n/m\rfloor}(1)$ be the Jordan block of size $\lfloor n/m\rfloor$, for some $m$, and eigenvalue 1 and define
 \begin{equation*}
	A=\matr{{c|c}\rho J_{\lfloor n/m\rfloor}(1)&0\\\hline0&I_{n-\lfloor n/m\rfloor}},\, H=\matr{{cc|ccc}e_1&\rho e_{\lfloor n/m\rfloor}&e_{\lfloor n/m\rfloor+1}&\cdots&e_n}.
\end{equation*}
Notice that we reduced the size of the chain by $1/m$ and we added $n-\lfloor n/m\rfloor$ directly excited states. By increasing $m$, we can achieve a larger distance to uncontrollability (constant smaller than $8$). However, we will still have exponential sample complexity of the order of at least $\lfloor n/m\rfloor$, based on the length of the chain.
\end{remark}
 \end{remark}
 
\section{Controllability index affects learnability}\label{sec:guaranteed}
Structural system properties of an underactuated system, such as the chained structure in the dynamics, can be critical in making system identification easy or hard.  This poses novel questions about understanding how system theoretic properties affect system learnability as defined in Definitions~\ref{def:poly} and~\ref{def:exp}.
We begin a new line of inquiry by characterizing how the controllability index $\kappa$, a critical structural system property,  affects the statistical properties of system identification. A brief review of the concept of controllability index can be found in the Appendix. It can be viewed as a structural measure of whether a system is directly actuated or underactuated resulting in long chains. The following theorem, is the first result connecting the controllability index with
sample complexity bounds.

\begin{theorem}[Controllability index-dependent upper bounds]\label{thm:upper_exponential}
Consider the set $\CC_n$ of systems $S=(A,B,H)$ such that Assumption~\ref{ass_boundedness} is satisfied. Let Assumption~\ref{ass:away_from_uncontrollability} be satisfied for the all pairs $(A,\matr{{cc}H&B})$.  Furthermore assume that  the controllability index of all pairs $(A,\matr{{cc}H&B})$ in the class is upper bounded by $\kappa$. Then, under the least squares system identification algorithm and white noise inputs $u_k~\sim\N(0,I_p)$, we obtain that
\begin{align}
		&\sup_{S\in\CC_n}\P_{S}(\snorm{A-\hat{A}_N}\ge \epsilon)\le \delta \nonumber
	\end{align}
is satisfied for
	\[
	N\ge    \poly^{\kappa}(n,M)\poly(\epsilon^{-1},\log 1/\delta).
	\]
\end{theorem}
Theorem~\ref{thm:upper_exponential} formalizes our intuition since the controllability index is the length of the chain from input excitation towards the most distant state in the chain.  Hence, systems with a large number of inputs (or noise) and small controllability index ($\kappa<<n$) are easy to identify. The directly excited case with isotropic noise, presented in Theorem~\ref{thm:directly_excited}, is a special case corresponding to a controllability index $\kappa=1$, recovering prior polynomial bounds. 

The implications of Theorems~\ref{thm:lower_exponential},~\ref{thm:upper_exponential} illustrate the impact controllability properties have on system learnability--see Figure~\ref{Fig:diagram}. 
Classes of systems with small controllability index~$O(1)$ have polynomial sample complexity. Classes where the index grows linearly $\Omega(n)$ can be exponentially hard in the worst case in general. There might still be subclasses of systems with large controllability indexes which nonetheless can be identified with a polynomial number of samples. However, we cannot provide any guarantees without further assumptions.

The proof of Theorem~\ref{thm:upper_exponential} crucially depends on the following system theoretic result that bounds the least singular value of the controllability Gramian (a quantitative measure of controllability) with the controllability index (a structural measure of controllability).  
\begin{theorem}[Controllability gramian bound]\label{thm:upper_bound_control_energy}
	Consider a system $(A,H)$ that satisfies Assumptions~\ref{ass_boundedness},~\ref{ass:away_from_uncontrollability}. Let $\kappa$ be its controllability index. Then, the least singular value of the gramian $\Gamma_{\kappa}$ is lower bounded by:
	\[
	\sigma_{\min}^{-1}(\Gamma_{\kappa})\le \poly^{\kappa}(M/\mu).
	\]
\end{theorem}
The above theorem is of independent interest, since it states that the controllability index rather than the dimension $n$ controls how fast the controllability Gramian degrades. While the above bound may be loose in general, it gives us qualitative insights about how system structure affects the hardness of input excitation and system identification. 
Our proof exploits the so-called ``staircase" (or Hessenberg) canonical representation~\eqref{eq:canonical_form} of state space systems~\cite{Dooren03}--see Appendix. The main idea is that if a system is robustly controllable then the coupling between the states is bounded away from zero. Hence, we can avoid the essentially uncontrollable systems of Theorem~\ref{thm:trivial_example} which lead to infinite sample complexity.
\section{Simulations}\label{sec:simulations}
We study three simulation scenarios to illustrate the qualitative implications of our results. In the first two cases, we verify that the sample complexity of the least squares algorithm can indeed grow exponentially with the dimension. In the third case, we investigate how the controllability index affects the sample complexity. In all cases, we perform Monte Carlo simulations to compute the empirical mean error $\snorm{A-\hat{A}_N}_2$ and we count the number of samples required to have error less than $\epsilon$, for some $\epsilon>0$. For numerical stability in the least squares estimator~\eqref{eq:least_squares} we used a regularization term (ridge  regression) with coefficient $0.001$.

In the first example in Section~\ref{sec:formulation}, Figure~\ref{Fig:motivational_example}, we used $1000$ Monte Carlo iterations to approximate the empirical average. We modeled the noise as gaussian with $w_k\sim\N(0,0.5)$ and used white noise inputs $u_k\sim\N(0,10)$. The sample complexity of the least squares algorithm seems to be exponential with the dimension. In Section~\ref{sec:hard}, we showed that such systems exhibit exponential sample complexity due to the weak coupling between the states.

In the second example, we study the behavior of Jordan blocks actuated from the last state. Let $J_n(\lambda)$ be a Jordan block of dimension $n$ and eigenvalues all $\lambda$. We consider the system $A=J_n(\lambda)$, $H=0.1e_n$, $B=5e_n$, which means we excite directly only state $x_{t,n}$. We repeat the same experiment as before for $1000$ Monte Carlo simulations with $w_k,u_k\sim\N(0,1)$ and for $\epsilon=0.005$. In Figure~\ref{Fig:jordan_block}, it seems that the complexity of the least squares algorithm is also exponential when $0<\lambda<1$. In this case the coupling between the states is not weak. However, certain subspaces might still be hard to excite. As $\lambda$ approaches the unit circle eigenvalue $1$ the complexity improves. For $\lambda=1$, after $n=9$ Matlab returned inaccurate results
as the condition number of the data becomes very large. Hence, we do not report any results beyond $n=9$. However, based on simulations for small $n$ it might be possible that the system can be learned by only a polynomial number of samples.
The intuition might be that in this case instability helps with excitation~\cite{simchowitz2018learning}. It is an open problem to prove or disprove exponential lower bounds for the Jordan block when $0<\lambda<1$. Similarly, we leave it as an open problem to prove or disprove polynomial upper bounds for the Jordan block when $\lambda=1$. 
\begin{figure}[h] \centering
	\definecolor{mycolor1}{rgb}{0.00000,0.44700,0.74100}%
\definecolor{mycolor2}{rgb}{0.85000,0.32500,0.09800}%
\definecolor{mycolor3}{rgb}{0.92900,0.69400,0.12500}%
\definecolor{mycolor4}{rgb}{0.49400,0.18400,0.55600}%
\resizebox{0.55\columnwidth}{!}{\begin{tikzpicture}
\begin{axis}[%
		width=6.028in,
		height=2.754in,
		at={(1.011in,0.642in)},
		scale only axis,
		xmin=5,
		xmax=13,
		tick label style={font=\LARGE},
			ylabel style={font=\huge},
	xlabel style={font=\huge},
		ymode=log,
		ymin=10,
		ymax=1000,
		xtick={5,6,7,8,9,10,11,12,13},
		yminorticks=true,
		xlabel={dimension $n$},
		ylabel={samples $N$},
		xmajorgrids,
		ymajorgrids,
		axis background/.style={fill=white},
		legend style={legend cell align=left, align=left, draw=white!15!black,font=\huge,legend pos=north west}
		]
		\addplot [color=mycolor1,line width=1.5pt]
		table[row sep=crcr]{%
5   65\\
6   88\\
7	125\\
8	171\\
9	236\\
10	326\\
11	480\\
12	700\\
13	999\\
		};
		\addlegendentry{$\lambda$=0.5}
			\addplot [color=mycolor2, dashed,line width=2pt]
		table[row sep=crcr]{%
5   53\\
6   71\\
7	96\\
8	134\\
9	181\\
10	254\\
11	357\\
12	507\\
13	731\\
		};
		\addlegendentry{$\lambda$=0.6}
			\addplot [color=mycolor3,dashdotted,line width=2pt]
		table[row sep=crcr]{%
5   41\\
6   52\\
7	68\\
8	92\\
9	115\\
10	147\\
11	194\\
12	276\\
13	368\\
		};
		\addlegendentry{$\lambda$=0.7}
					\addplot [color=mycolor4,dotted,line width=2pt]
		table[row sep=crcr]{%
5   25\\
6   31\\
7	36\\
8	40\\
9	45\\
		};
		\addlegendentry{$\lambda$=1}
	\end{axis}
	\end{tikzpicture}}
 	\caption{Sample complexity of identifying the Jordan block of size $n$ and eigenvalues all $\lambda$, actuated from the last state. The figure shows the minimum number of samples $N$ such that the (empirical) average error $\mathbb{E}\snorm{A-\hat{A}_N}_2$ is less than $0.005$. The sample complexity appears to be increasing exponentially with the dimension $n$ for $\lambda<1$.  For $\lambda=1$, Matlab returns inaccurate results for $n\ge 10$ since the condition number of the data is very large. However, in the regime $5\le n\le 9$, the complexity seems to be polynomial, increasing in 5 sample increments. }
	\label{Fig:jordan_block}
\end{figure}
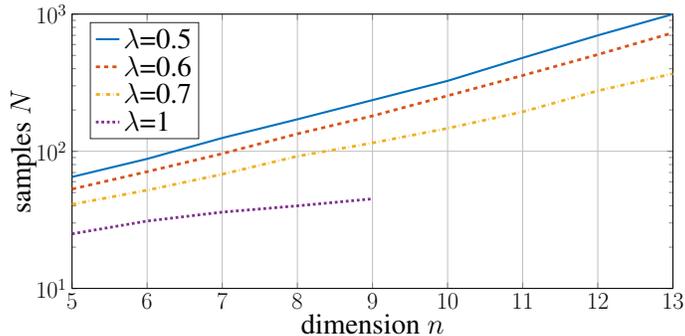

In the third example, we consider the Jordan block $A=J_n(0.5)$ with noise $H=0.1e_n$. We start from $B=5e_n$ and we gradually add more exogenous inputs to decrease the controllability index: we try $B=5\matr{{cc}e_{n}&e_{\lceil n/2\rceil}}$ and $B=5\matr{{ccc}e_n&e_{n-2}&\dots}$ which correspond to indices $\kappa=\lceil n/2\rceil$ and $\kappa=2$ respectively. We repeat the same experiment as before for $1000$ Monte Carlo simulations with $w_k,u_k\sim\N(0,1)$ and for $\epsilon=0.005$. In Figure~\ref{Fig:contr_index}, it seems that the sample complexity remains exponential when $\kappa=\lceil n/2\rceil$. However, when $\kappa=2$ there is a phase transition and the sample complexity becomes polynomial with the dimension.

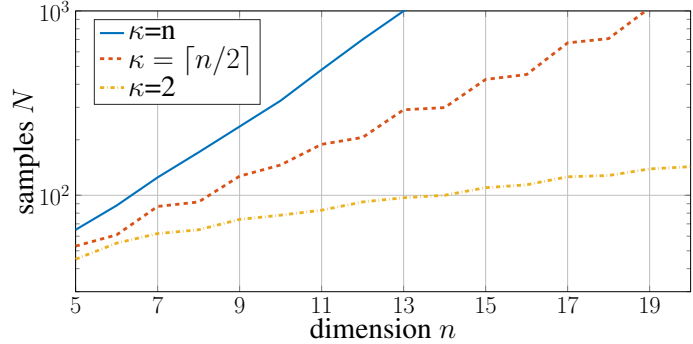
\begin{figure}[t] \centering
	\definecolor{mycolor1}{rgb}{0.00000,0.44700,0.74100}%
\definecolor{mycolor2}{rgb}{0.85000,0.32500,0.09800}%
\definecolor{mycolor3}{rgb}{0.92900,0.69400,0.12500}%
\resizebox{0.55\columnwidth}{!}{\begin{tikzpicture}
\begin{axis}[%
		width=6.028in,
		height=2.754in,
		at={(1.011in,0.642in)},
		scale only axis,
		xmin=5,
		xmax=20,
		tick label style={font=\LARGE},
			ylabel style={font=\huge},
	xlabel style={font=\huge},
		ymode=log,
		ymin=30,
		ymax=1000,
		xtick={5,7,9,11,13,15,17,19},
		xlabel={dimension $n$},
		ylabel={samples $N$},
		xmajorgrids,
		ymajorgrids,
		axis background/.style={fill=white},
		legend style={legend cell align=left, align=left, draw=white!15!black,font=\huge,legend pos=north west}
		]
		\addplot [color=mycolor1,line width=1.5pt]
		table[row sep=crcr]{%
5   65\\
6   88\\
7	125\\
8	171\\
9	236\\
10	326\\
11	480\\
12	700\\
13	999\\
14  1439\\
15  2116\\
		};
		\addlegendentry{$\kappa$=n}
			\addplot [color=mycolor2, dashed,line width=2pt]
		table[row sep=crcr]{%
5   53\\
6   61\\
7	87\\
8	92\\
9	127\\
10	146\\
11	189\\
12	206\\
13	291\\
14 299\\
15 425\\
16 452\\
17 671\\
18 708\\
19 1038\\
20 1151\\
		};
		\addlegendentry{$\kappa=\lceil n/2\rceil$}
			\addplot [color=mycolor3,dashdotted,line width=2pt]
		table[row sep=crcr]{%
5   45\\
6   55\\
7	62\\
8	65\\
9	74\\
10	78\\
11	83\\
12	92\\
13	97\\
14 100\\
15 110\\
16 114\\
17 126\\
18 128\\
19 139\\
20 143\\
		};
		\addlegendentry{$\kappa$=2}
	\end{axis}
	\end{tikzpicture}}
 	\caption{Sample complexity of identifying the Jordan block $J_n(0.5)$ of size $n$ and eigenvalues all $0.5$, for different values of the controllability index. The figure shows the minimum number of samples $N$ such that the (empirical) average error $\mathbb{E}\snorm{A-\hat{A}_N}_2$ is less than $0.005$. The sample complexity appears to be increasing exponentially with the dimension $n$ for $\kappa=\Theta(n)$.  For $\kappa=2$, the sample complexity is much smaller and increases polynomially.}
	\label{Fig:contr_index}
\end{figure}

\section{Conclusion}
The results of this paper paint a broader and more diverse landscape about the statistical complexity of learning linear systems, summarized in Figure~\ref{Fig:diagram} according to the controllability index $\kappa$ of the considered system class.
While statistically easy cases that were previously known are captured by Theorem~\ref{thm:directly_excited}, we also showed that hard system classes exist (Theorem~\ref{thm:lower_exponential}). By exploiting structural system theoretic properties, such as the controllability index,  we broadened the class of easy to learn linear systems (Theorem~\ref{thm:upper_exponential}).
\begin{figure}[h] \centering{
		\includegraphics[width=0.6\columnwidth]{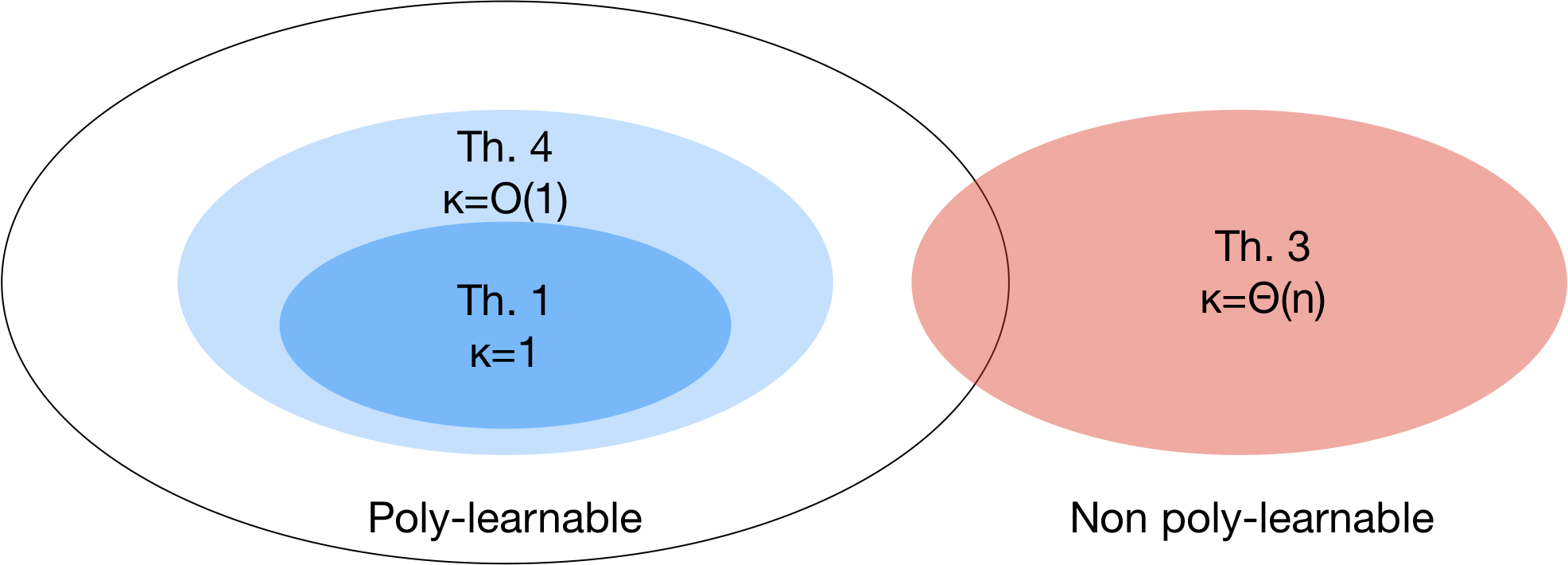}}   
	\caption{Sample complexity classes for linear systems.
	according to their controllability index. 
	}
	\label{Fig:diagram}
\end{figure}

 Our results pose numerous future questions for exploiting other system properties (e.g. observability) for efficiently learning classes of partially-observed linear systems or nonlinear systems. 
It remains an open problem to prove whether or not the $n-$th order integrator is poly-learnable as discussed in Section~\ref{sec:simulations}. Similarly, it is an open problem to prove whether or not the Jordan block of size $n$ and eigenvalues all $0<\lambda<1$ has exponential complexity. Finally, the results of this paper might have ramifications for control, for example learning the linear quadratic regulator, as well as reinforcement learning.

\bibliographystyle{IEEEtran}
\bibliography{IEEEabrv,Literature_Identification}
\counterwithin{lemma}{section}
\counterwithin{theorem}{section}
\counterwithin{proposition}{section}
\counterwithin{corollary}{section}
\counterwithin{equation}{section}
\appendix
\section{Controllability-related concepts}
We briefly review the concept of controllability and other related concepts. We consider the pair $(A,H)$, but the same definitions hold also for $(A,B)$.
The controllability matrix of $(A,H)$ is defined as 
\begin{align*}
	\C_k(A,H)\triangleq \matr{{cccc}H&AH&\cdots&A^{k-1}H},\,k\ge 1.
\end{align*}The pair $(A,H)$ is \emph{controllable} when the controllability matrix $\C_{n}(A,H)$
has full column rank $n$. 
The \emph{controllability Gramian} at time $k$ is defined as :
\begin{align*}
	\Gamma_k(A,H)\triangleq \C_{k}(A,H)\C'_{k}(A,H)=\sum_{i=0}^{k-1}A^{i}HH'(A')^{i}.
\end{align*}

If $H$ is not a column matrix, the full column rank condition might be satisfied earlier for some $k\le n$. The minimum time that we achieve controllability is the \emph{controllability index}:
\begin{align}\label{eq:controllability_idex}
	\kappa(A,H)\triangleq \min\set{ k\ge 1: \rank(\C_k(A,H))=n }.
\end{align}
It is the lag between the time the disturbance $w_t$ is applied and the time $t+\kappa$ by which we see the effect of that disturbance in all states. This lag is non-trivial if the number of disturbances $r<n$ is smaller than the number of states; in this case we call the system underactuated.

Based on the fact that the rank of the controllability matrix at time $\kappa$ is $n$, we can show that the pair $(A,H)$ admits the following canonical representation, under a unitary similarity transformation~\cite{Dooren03}. 
\begin{proposition}[Staircase form]\label{prop:canonical}
	Consider a controllable pair $(A,H)$ with controllability index $\kappa$ and controllability matrix $\C_k$, $k\ge 0$. There exists a unitary similarity transformation $U$ such that $U'U=UU'=I$ and:
	\begin{equation}
	\begin{aligned}\label{eq:canonical_form}
		U'H&=\matr{{cccc}H'_1&0&\cdots&0}'
		\\
		U'AU&=\matr{{ccccc}A_{1,1}&A_{1,2}&\cdots&A_{1,\kappa-1}&A_{1,\kappa}\\A_{2,1}&A_{2,2}&\cdots&A_{3,\kappa-1}&A_{2,\kappa}\\0&A_{3,2}&\cdots&A_{3,\kappa-1}&A_{3,\kappa}\\0&0&\cdots&A_{4,\kappa-1}&A_{4,\kappa}\\\vdots& & &\vdots&\\0&0&\cdots&A_{\kappa,\kappa-1}&A_{\kappa,\kappa}},
	\end{aligned}
	\end{equation}
	where $A_{i,j}\in \R^{r_i\times r_j}$ are block matrices, with $r_i=\rank(\C_{i})-\rank(\C_{i-1})$, $r_1=r$, $H_1\in\R^{r\times r}$. Moreover, the matrices $A_{i+1,i}$ have full row rank $\rank(A_{i+1,i})=r_{i+1}$ and the sequence $r_i$ is decreasing.
\end{proposition}
The above representation is useful as it captures the coupling between the several sub-states via the matrices $A_{i+1,i}$. If one of these matrices $A_{i+1,i}$ is close to zero then the system will be close to being uncontrollable. On the other hand, if a system is robustly controllable then these matrices are bounded away from being row-rank deficient. Since the similarity transformation is unitary it does not affect properties of the system like the minimum singular value of the controllability Gramian.
The proof of Theorem~\ref{thm:upper_bound_control_energy} exploits the above ideas--see Section~\ref{sec:proof_Gramian} for more details.

\section{Distance from uncontrollability properties}
In this section we review properties of the distance from uncontrollability. The main focus is to prove that standard systems, like the integrator, have  distance to uncontrollability which degrades linearly with the dimension $n$. 
\begin{lemma}\label{lem:integrator_distance_to_uncontrollability}
Let $0<\rho<1$ and consider the perturbed $n-$th order integrator:
\[
A= \rho\matr{{cccccc}1 &1&0&\cdots&0&0\\0& 1&1&\cdots&0&0\\& & &\ddots&\\0&0&0&\cdots&1&1\\0&0&0&\cdots&0&1},\, H=\rho\matr{{c}0\\\vdots\\1} 
\]
The distance from uncontrollability is given by
\begin{equation}\label{eq:integrator_distance}
d(A,H)=\rho\sin\paren{\frac{\pi}{n+1}}.
\end{equation}
As a result the distance degrades linearly:
\begin{equation}\label{eq:distance_degradation_linear}
\rho\frac{2}{n+1} \le d(A,H) \le \rho\frac{\pi}{n+1},
\end{equation}
for $n\ge 1$.
\end{lemma}

\begin{proof}
The proof follows from the fact that the distance form uncontrollability is equivalently given by the formula~\cite{eising1984metric}:
\begin{equation}\label{eq:metric_formula}
    d(A,H)=\inf_{s\in\mathbb{C}}\sigma_{\min}(\matr{{cc}A-sI&H}), 
\end{equation}
and results about the eigenvalues of Toeplitz matrices~\cite{tridiagonalToeplitz1999}.

In more detail, let $^*$ denote the complex conjugate. We have: 
\[\matr{{cc}A-sI&H}\matr{{cc}A-sI&H}^*= \T_s,\]
where
\begin{equation}\label{eq:Toeplitz_distance}
    \T_s=\matr{{ccccc}\abs{\rho-s}^2+\rho^2 &\rho(\rho-s^*)& &0&0\\\rho(\rho+s^*)& \abs{\rho-s}^2+\rho^2& &0&0\\& &\ddots&\\0&0& &\abs{\rho-s}^2+\rho^2&\rho(\rho-s^*)\\0&0& &\rho(\rho+s^*)&\abs{\rho-s}^2+\rho^2}
\end{equation}
is a tri-diagonal Toeplitz matrix, with all diagonal elements equal to $\abs{\rho-s}^2+\rho^2$, all superdiagonal elements equal to $\rho(\rho-s^*)$ and subdiagonal elements equal to $\rho(\rho+s^*)$. Based on~\cite[Th 2.2]{tridiagonalToeplitz1999}, the smallest eigenvalue of $\T$ is equal to:
\[
\sigma_{\min}(\T_s)=\abs{\rho-s}^2+\rho^2-2\abs{\rho}\abs{\rho-s}\cos(\pi/(n+1)).
\]
The above quantity is minimized for $\hat{s}=\rho+\abs{\rho} \cos(\pi/(n+1))$. Hence, we can compute the distance to uncontrollability:
\[
d(A,H)=\sqrt{\sigma_{\min}(\T_{\hat{s}})}=\abs{\rho}\sin(\pi/(n+1)).
\]
Finally~\eqref{eq:distance_degradation_linear} follows from~\eqref{eq:integrator_distance} using the elementary calculus inequality
\[\frac{2x}{\pi}\le\sin x\le x,\,\text{for }0\le x\le \pi/2,\]
which completes the proof.
\end{proof}

\begin{lemma}\label{lem:integrator_like_distance_to_uncontrollability}
System~\eqref{eq:difficult_system}
is $\mu$-bounded away from uncontrollability with $\mu^{-1}\le \rho^{-1}(n+1)$.
\end{lemma}
\begin{proof}

Let $^*$ denote the complex conjugate. Then we have: 
\[\matr{{cc}A-sI&H}\matr{{cc}A-sI&H}^*= \T_s+e_1e_1'\succeq \T_s\]
where $\T_s$ is a tridiagonal Toeplitz matrix defined above in~\eqref{eq:Toeplitz_distance}. 
Now the proof is identical to the proof of Lemma~\ref{lem:integrator_distance_to_uncontrollability} but we have inequality instead of equality:
\[
d(A,H)\ge \sqrt{\sigma_{\min}(\T_{\hat{s}})}=\abs{\rho}\sin(\pi/(n+1))\ge 2\abs{\rho}/(n+1)\ge \abs{\rho}/(n+1).
\]
\end{proof}

\begin{lemma}[Triangle inequality]\label{lem:distance_to_uncontrollability_triangle}
Let $d(A,H)$ be the distance to uncontrollability for some matrices $A\in\R^{n\times n},H\in\R^{r\times n}$ and let $\snorm{\hat{A}-A}_2\le \epsilon<d(A,H)$ for some matrix $\hat{A}\in\R^{n\times n}$. Then:
\begin{equation}\label{eq:controllability_distance_triangle_inequality}
d(\hat{A},H)\ge d(A,H)-\epsilon.
\end{equation}
\end{lemma}
\begin{proof}
Assume that $d(\hat{A},H)< d(A,H)-\epsilon$ and let $\matr{{cc}\Delta \hat{A}&\Delta \hat{H}}$ be the perturbation such that $(\hat{A}+\Delta \hat{A},H+\Delta \hat{H})$ is uncontrollable with $d(\hat{A},H)=\snorm{\matr{{cc}\Delta \hat{A}&\Delta \hat{H}}}_2$. Then, we can define a perturbation for the original pair $(A,H)$ that contradicts the definition of $d(A,H)$:
\[
\Delta A=A-\hat{A}+\Delta \hat{A},\,\Delta H=\Delta \hat{H}.
\]
The perturbation makes $(A,H)$ uncontrollable and by the triangle inequality, it has norm $\snorm{\matr{{cc}\Delta A&\Delta H}}_2\le d(\hat{A},H)+\epsilon<d(A,H)$. Since this is impossible~\eqref{eq:controllability_distance_triangle_inequality} holds. 
\end{proof}

 \section{Proof of Lemma~\ref{lem:powers_of_A}}
	In this section, we establish upper bounds on the gramian matrices $\Gamma_k$. Contrary to previous approaches we avoid using the Jordan form of matrix $A$. We do not want our bounds to depend on the condition number of the Jordan transformation which can be ill-posed and badly conditioned. Instead, we should use stable transformations like the Schur decomposition.
	\begin{proof}
	When $n=1$ the proof is immediate. So let $n\ge 2$.
	Consider the Schur triangular form~\cite[Chapter 2.3]{horn2012matrix} of $A$:
	\[
	A=UDU^{*},
	\]
	where $D$ is upper triangular, $U$ is unitary, and $*$ denotes complex conjugate. Let $\Lambda$ be the diagonal part of $D$, which contains all eigenvalues of $A$ as elements. 
Notice that $D-\Lambda$ is upper triangular with zero diagonal elements, while $\Lambda$ is diagonal. Thus, any product of the form \[\Lambda^{t_0}(D-\Lambda)^{s_1}\Lambda^{t_1}\cdots (D-\Lambda)^{s_k}\Lambda^{t_k}=0,\text{ if }s_1+\cdots+s_k\ge n.\]
where $s_1,\dots,s_k$ and $t_0,t_1,\dots,t_k$ are two collections of integers, for some $k\ge 1$. 
Now we can simplify the expression:
\begin{align*}
D^{k}&=(\Lambda+D-\Lambda)^k=\sum_{d_1,\dots,d_k\in\set{0,1}^k} F_{d_1}\cdots F_{d_k}\\
&=\sum_{\substack{d_1,\dots,d_k\in\set{0,1}^k\\d_1+\cdots+d_k\le n-1}} F_{d_1}\cdots F_{d_k},
\end{align*}
where $F_{1}=D-\Lambda$, $F_0=\Lambda$.
Notice that $\norm{D-\Lambda}_2\le \norm{D}_2=\norm{A}_2\le M$, where the first inequality follows from the fact that $D-\Lambda$ is a submatrix if $D$. Since the eigenvalues of $A$ are inside or on the unit circle, we have $\norm{\Lambda^{t}}_2\le 1$, for all $t\ge 0$.  Hence, by a counting argument
\begin{align*}
\norm{A^k}_2&=\norm{D^k}\le \sum_{t=0}^{n-1}{\binom{k}{t}}\max\set{M^{t},1}\\
&\le \sum_{t=0}^{n-1}{\binom{k}{t}}\max\set{M^{n-1},1}.
\end{align*}
To conclude, we use the known bound~\cite[Exercise 0.0.5]{vershynin2018high}:
\[
\sum_{t=0}^{n-1}{\binom{k}{t}}\le \paren{\frac{ek}{n-1}}^{n-1}
\]
	\end{proof}
	Since we obtained a bound on the powers of matrix $A$, we can immediately obtain an upper bound on the Gramian as a corollary.
	\begin{corollary}
	Let $A\in\R^{n\times n}$ have all eigenvalues inside or on the unit circle, with $\norm{A}_2\le M$. Let $H\in\R^{n\times r}$, $r\le n$ with $\norm{H}_2\le M$. Then, the gramian $\Gamma_k(A,H)$ is upper bounded by:
\begin{equation}\label{eq:gramian_upper_bound}
    \norm{\Gamma_{k}(A,H)}_2\le e^{2n-2}k^{2n-1}\max\set{M^{2n},1}
\end{equation}
	\end{corollary}

\section{Proof of Theorem~\ref{thm:trivial_example}}
	Let $\beta$ be any non-zero number. Fix an accuracy parameter $\epsilon>0$ and a confidence $0<\delta<1$. Consider the systems:
	\begin{align*}
	A_1&=\matr{{ccc}0&0&0\\0&0&\beta\\0&0&0},\,A_2=\matr{{ccc}0&2\epsilon&0\\0&0&\beta\\0&0&0},\\H_1&=H_2=\matr{{cc}e_1&e_3}.
	\end{align*}
	Both systems are controllable and belong to the class $\C_{n}$ for any non-zero $\beta\neq 0$. However, they are arbitrarily close to uncontrollability for small $\beta$. Let $f_{S_i}(x_0,\dots,x_N)$ denote the probability density function of the distribution of the data under system $S_i$, $i=1,2$. Then the log-likelihood ratio under $S_1,S_2$ is:
	\[
	L_{N}=\log\frac{f_{S_1}(x_0,\dots,x_{N})}{f_{S_2}(x_0,\dots,x_{N})}.
	\]
	Due to the Markovian structure of the linear system, we can write $f_{S_i}(x_0,\dots,x_N)=\prod_{k=1}^{N} f_{S_i}(x_{k}|x_{k-1})$, for $i=1,2$. 
	Moreover, due to the structure of the dynamical systems: 
\[f_{S_i}(x_{k}|x_{k-1})=f_{S_i}(x_{k,1}|x_{k-1,2})f_{S_i}(x_{k,2},x_{k,3}|x_{k-1,2}x_{k-1,3}).\]
However, systems $A_1,A_2$ have identical distributions for $x_{k,2}$ and $x_{k,3}$. As a result, the log-likelihood ratio becomes:
		\[
	L_{N}=\sum_{k=1}^{N}\log\frac{f_{S_1}(x_{k,1}|x_{k-1,2})}{f_{S_2}(x_{k,1}|x_{k-1,2})}.
	\]
	The KL divergence can now be computed:
	\begin{align*}
	&\mathrm{KL}(\P_{S_1},\P_{S_2})=\E_{S_1}L_N\\
	&=\E_{S_1}\sum_{k=1}^{N}\E_{S_1}\paren{\log\frac{f_{S_1}(x_{k,1}|x_{k-1,2})}{f_{S_2}(x_{k,1}|x_{k-1,2})}|\F_{k-1}}\\
	&=\E_{S_1}\sum_{k=1}^{N}\mathrm{KL}(\N(0,1),\N(2\epsilon x_{k-1,2},1))\\
	&=\E_{S_1}\sum_{k=1}^{N}(2\epsilon x_{k-1,2})^2/2\le 2\epsilon^2 N \Gamma_{N,22}(A,H),
	\end{align*}
	where we used $\E_{S_1}x^2_{k-1,2}=\Gamma_{k-1,22}\le \Gamma_{N,22}$ along with the fact that the KL-divergence between two scalar Gaussians is:
	\[
	\mathrm{KL}(\N(\mu_1,1),\N(\mu_2,1))=(\mu_1-\mu_2)^2/2
	\]
	
	A simple computation shows that  $\Gamma_{k,22}=b^2$, for all $k\ge 1$.
	Then, it follows from Lemma~\ref{lem:minimax} that~\eqref{eq:objective} holds only if:
	\[
 N\ge \frac{1}{\beta^2 2\epsilon^2}\log\frac{1}{3\delta }.
	\]
	However $\beta$ is arbitrary, which implies that~\eqref{eq:objective} holds only if:
	\[
	N\ge \sup_{\beta\neq 0}\frac{1}{\beta^2 4\epsilon^2}\log\frac{1}{3\delta }=\infty.
	\]

\section{Proof of Theorem~\ref{thm:lower_exponential}}
	Consider system~\eqref{eq:difficult_system} with $\rho=1/4$ and the perturbed system $\tilde{A}=A+2\epsilon e_1 e'_2$, $\tilde{H}=H$, where we modify $A_{12}$ by $2\epsilon$.
	Both pairs $(A,H)$, $(\hat{A},\hat{H})$ are controllable. From Lemma~\ref{lem:integrator_like_distance_to_uncontrollability}, we obtain that $d(A,H)\ge (4(n+1))^{-1}\ge (8(n+1))^{-1}$.
	Fix an $\epsilon\le(16(n+1))^{-1}$. Then, from Lemma~\ref{lem:distance_to_uncontrollability_triangle}, we also get that $d(\hat{A},\hat{H})\ge d(A,H)-2\epsilon\ge(8(n+1))^{-1}$. 
	Hence, both systems belong to the class $\CC_n$. 
	
	Define $S_1=(A,H)$, $S_2=(\hat{A},\hat{H})$. Following the same arguments as in the proof of Theorem~\ref{thm:trivial_example}, the KL divergence of the distribution of the data under $A$ and $\hat{A}$ is equal to
	\begin{align*}
	&\mathrm{KL}(\P_{S_1},\P_{S_2})=\E_{S_1}L_N\\
	&=\E_{S_1}\sum_{k=1}^{N}\E_{S_1}\paren{\log\frac{f_{S_1}(x_{k,1}|x_{k-1,2})}{f_{S_2}(x_{k,1}|x_{k-1,2})}|\F_{k-1}}\\
	&=\E_{S_1}\sum_{k=1}^{N}\mathrm{KL}(\N(\rho x_{k-1,2},1),\N((\rho+2\epsilon) x_{k-1,2},1))\\
	&=\E_{S_1}\sum_{k=1}^{N}(2\epsilon x_{k-1,2})^2/2\le 2\epsilon^2 N \Gamma_{N,22}(A,H).
	\end{align*}
	From Lemma~\ref{lem:difficult_system_gramian}, we obtain the exponential decay bound:
	\[
	\Gamma_{N,22}(A,H)\le 4^{-n+2}/3.
	\]
	Finally, from Lemma~\ref{lem:minimax}, equation~\eqref{eq:objective} holds only if:
	\[
	N\ge \frac{1}{2\epsilon^2 \Gamma_{N,22}(A,H)}\log\frac{1}{3\delta}\ge \frac{4^{n-2}}{6 \epsilon^2}\log\frac{1}{3\delta}.
	\]

	\begin{lemma}\label{lem:difficult_system_gramian}
Consider system~\eqref{eq:difficult_system} with $\rho<1/2$. Then 
\[
\Gamma_{k,22}(A,H)\le (2\rho)^{2n-2}/(1-4\rho^2).
\]
\end{lemma}
\begin{proof}
Notice that $e_2'A^sH=0$ for all $s\le n-2$ and $\norm{A}\le 2\rho<1$. Hence,
\begin{align*}
e'_2\Gamma_k(A,H) e_2&\le \sum_{s=n-1}^{k}e_2A^{s}QA^{'s}e_2'\\
&\le \sum_{s=n-1}^{\infty}(2\rho)^{2s}=(2\rho)^{2n-2}/(1-4\rho^2).
\end{align*}
\end{proof}
\section{Proof of Theorem~\ref{thm:upper_exponential}}
By $\Gamma_k=\Gamma_k(A,H)+\Gamma_k(A,B)$ we denote the Gramian under both $H,B$. Define also the sigma-algebra:
\[
\bar{\F}_k=\sigma(w_0,u_0,\dots,w_k,u_k).
\]
We will apply Theorem~2.4 in~\cite{simchowitz2018learning} to the combined state-input vectors with three modifications since the noise is not isotropic. First, we compute the sub-Gaussian parameter of the noise.
\begin{defin}
A zero mean random vector $w\in\R^{r\times 1}$ is called $\sigma^2-$sub-Gaussian with respect to a sigma algebra $\F$ if for every unit vector $u\in\R^{r\times}$:
\[
\E \paren{e^{su'w}|\F}\le e^{s^2\sigma^2/2}.
\]
\end{defin}
From the definition, it follows that the non-isotropic Gaussian vector $Hw_k$ is sub-Gaussian with parameter $\norm{H}^2_2$.
\begin{lemma}
Let $w_k\in\R^{r\times 1}$ be 1-sub-Gaussian with respect to $\bar{\F}_{k-1}$. Then $Hw_k$ is $\norm{H}^2_2-$sub-Gaussian with respect to $\bar{\F}_{k-1}$.
\end{lemma}
\begin{proof}
Let $u\in\R^{r\times 1}$ be a unit vector. Then:
\begin{align*}
&\E \paren{e^{su'Hw_k}|\bar{\F}_{k-1}}=\E \paren{e^{s\norm{u'H}\frac{u'H}{\norm{u'H}}w_k}|\bar{\F}_{k-1}}\\
&\le e^{s^2\norm{u'H}^2_2/2}\le e^{s^2\norm{H}^2_2/2}
\end{align*}
\end{proof}
Second, define $y_{k}=\matr{{cc}x'_k&u'_k}'$. It follows that for all $j\ge 0$ and all unit vectors $v\in\R^{(n+p)\times 1},$ the following small-ball condition is satisfied:
\begin{equation}\label{eq:small_ball}
    \frac{1}{2\kappa}\sum_{t=0}^{2\kappa}\P(\abs{v'y_{t+j}}\ge \sqrt{v'\Gamma_{\mathrm{sb}}v} |\bar{\F}_j)\ge \frac{3}{20},
\end{equation}
where
\begin{equation}\label{eq:small_ball_covariance}
    \Gamma_{\mathrm{sb}}=\matr{{cc}\Gamma_{\kappa}&0\\0&I_p}.
\end{equation}
Equation~\eqref{eq:small_ball} follows from the same steps as in Proposition~3.1 in~\cite{simchowitz2018learning} with the choice $k=2\kappa$.

Finally, we determine an upper bound $\bar{\Gamma}$ for the gram matrix $\sum_{t=0}^{N-1}y_ty'_t$. Using a Markov inequality argument as in~\cite[proof of Th 2.1]{simchowitz2018learning}, we obtain that
\[
\P(\sum_{t=0}^{N-1}y_ty'_t\preceq \bar{\Gamma})\ge 1-\delta,
\]
where 
\[\bar{\Gamma}=\frac{n+p}{\delta}N\matr{{cc}\Gamma_N&0\\0&I_p}\]

Now we can apply Theorem 4.2 of~\cite{simchowitz2018learning}. With probability at least $1-3\delta$ we have $\snorm{A-\hat{A}_N}\le \epsilon$ if:
\begin{align*}
   N&\ge \frac{\poly(n,\log1/\delta,M)}{\epsilon^2\sigma_{\min}(\Gamma_{\kappa})}\log\det (\bar{\Gamma}\Gamma^{-1}_{\mathrm{\kappa}}),
\end{align*}
where we have simplified the expression by including terms in the polynomial term. Based on Lemma~\ref{lem:powers_of_A} and Theorem~\ref{thm:upper_bound_control_energy}, we can bound the right-hand side:
\begin{align*}
\frac{\poly(n,\log1/\delta,M)}{\epsilon^2\sigma_{\min}(\Gamma_{\kappa})}\log\det (\bar{\Gamma}\Gamma^{-1}_{\mathrm{\kappa}})&\le \poly(n,\epsilon^{-1},\log1/\delta,M)\poly\paren{\frac{M}{\mu}}^{\kappa}\log N\\
&\le \poly(n,\epsilon^{-1},\log1/\delta,M)\poly\paren{M,n}^{\kappa}\log N,
\end{align*}
where we used the fact that $\mu^{-1}\le \poly(n)$.
Hence, it is sufficient to have:
\[
N\ge \poly\paren{n,\epsilon^{-1},\log1/\delta,M}\poly\paren{M,n}^{\kappa}\log N.
\]
To obtain the final polynomial bound, we need to remove the logarithm of $N$. It is sufficient to apply the inequality:
\[
N\ge c\log N \text{ if }N\ge 2c\log 2c,
\]
for $c>0$
which follows from elementary calculus.
	\section{Proof of Theorem~\ref{thm:upper_bound_control_energy}}\label{sec:proof_Gramian}
	Our goal is to upper bound the norm of the Moore-Penrose pseudo-inverse $\snorm{\C^{\dagger}_{\kappa}}=\sqrt{\sigma_{\min}(\Gamma_{\kappa})}$, where the equality follows from the SVD decomposition and the definition of the gramian. Towards proving the result, we will work with the staircase form~\eqref{eq:canonical_form}.
	First, we show that if the system is $\mu$-away from uncontrollability, then the subdiagonal matrices in the staircase form are bounded away from zero.
\begin{lemma}[Staircase form lower bound]
Let $(A,H)\in\R^{n\times(n+r)}$ be controllable and let Assumption~\ref{ass:away_from_uncontrollability} hold. Consider the staircase form of $(A,H)$, with $A_{i+1,i}$ the subdiagonal matrices, for $i=1,\dots,\kappa-1$, where $\kappa$ is the controllability index. Then, we have  $A_{i+1,i}A'_{i+1,i}\succeq \mu^2I_{r_{i+1}}$ for all $i=1,\dots,\kappa-1$. Moreover, $H_1H_1'\succeq \mu^2 I_r$.
\end{lemma}
\begin{proof}
Let $(\hat{A},\hat{H})$ be the staircase form of $(A,H)$ under the unitary similarity transformation $U$. First, we show that the controllability metric is invariant to unitary transformations. Denote $\Delta \hat{A}=U^*\Delta A U$, $\Delta\hat{H}=U^*\Delta H$. Then:
\begin{align*}
&\min\set{\snorm{\matr{{cc}\Delta A&\Delta H}}_2:\:(A+\Delta A,H+\Delta H)\text{ unc.}}\\
    &=\min\set{\snorm{\matr{{cc}\Delta \hat{A}&\Delta \hat{H}}}_2:\:(A+\Delta A,H+\Delta H)\text{ unc.}}\\
    &=\min\set{\snorm{\matr{{cc}\Delta \hat{A}&\Delta \hat{H}}}_2:\:(\hat{A}+\Delta \hat{A},\hat{H}+\Delta \hat{H})\text{ unc.}}
\end{align*}
where the first equality follows from $\snorm{\matr{{cc}\Delta A&\Delta H}}_2=\snorm{\matr{{cc}U^*\Delta AU&U^*\Delta H}}_2$. The second equality follows from the fact that controllability is preserved under similarity transformations
As a result, $d(\hat{A},\hat{H})=d(A,H)\ge \mu$.

Note that $A_{i+1,i}\in\R^{r_{i+1}\times r_i}$. Hence, it is sufficient to show that $\sigma_{r_{i+1}}(A_{i+1,i})\ge \mu$, where $\sigma_{r_{i+1}}$ denotes the $r_{i+1}$ smallest singular value. Assume that the opposite is true $\sigma_{r_{i+1}}(A_{i+1,i})< \mu$. We will show that this contradicts the fact that $(\hat{A},\hat{H})$ is away from uncontrollability: $d(\hat{A},\hat{H})=d(A,H)\ge \mu$. Let $u$ and $v$ be the singular vectors in the Singular Value Decomposition of $A_{i+1,i}$ corresponding to $\sigma_{r_{i+1}}$. Let $\Delta A_{i+1,i}\triangleq -\sigma_{r_{i+1}}(A_{i+1,i})uv'$. Then $A_{i+1,i}+\Delta A_{i+1,i}$ is rank deficient. Now let $\Delta\hat{A}$ be zero everywhere apart from the block $\Delta A_{i+1,i}$. Then, we have that $(\hat{A}+\Delta \hat{A},\hat{H})$ is uncontrollable, with $\snorm{\Delta\hat{A}}_2<\mu\le d(\hat{A},\hat{H})$, which is impossible. The proof for $H_1$ is similar.
\end{proof}

The above result allows us to work with the staircase form~\eqref{eq:canonical_form}, which has a nice triangular structure. In fact the controllability matrix is block-triangular and we can upper-bound its least singular value using a simple recursive bound. Since the least singular value of the Gramian is invariant to similarity transformations, we will now assume that the system $(A,H)$ is now already in form~\eqref{eq:canonical_form} with $U=I$.
Let us define some auxiliary matrices that will help us prove Theorem~\ref{thm:upper_bound_control_energy}. With $\tilde{A}_{k}$, for $k\le \kappa$ we denote the submatrix of $A$  when we keep the $k$-upper left block matrices in~\eqref{eq:canonical_form} and we delete the remaining columns and rows, e.g.:
\[
\tilde{A}_{2}=\matr{{cc}A_{1,1}&A_{1,2}\\A_{2,1}&A_{2,2}},\tilde{A}_3=\matr{{ccc}A_{1,1}&A_{1,2}&A_{1,3}\\A_{2,1}&A_{2,2}&A_{2,3}\\0&A_{3,2}&A_{3,3}},\dots
\]
Similarly, we define the submatrices $\tilde{H}_k$ where we keep only the upper $k$ blocks of the matrix $H$:
\[
\tilde{H}_1=H_1,\,\tilde{H}_2=\matr{{cc}H'_1&0}',\dots.
\]
Finally, define the upper-left controllability submatrices $\tilde{\C}_{k}$:
\begin{equation}\label{eq:upper_left_controllability}
\tilde{\C}_{k}=\matr{{cccc}\tilde{H}_k&\tilde{A}_k\tilde{H}_k&\dots&\tilde{A}^{k-1}_k\tilde{H}_k}\in\R^{\sum_{i=1}^{k}r_i\times (kr)}.
\end{equation}
The benefit of working with the above matrices is that they are block upper-triangular. For example:
\[
\tilde{\C}_1=H_1,\,\tilde{\C}_2=\matr{{cc}H_1&A_{1,1}H_1\\0&A_{2,1}H_1},\dots
\]
By definition $\tilde{A}_{\kappa}=A$, $\tilde{H}_{\kappa}=H$, and $\tilde{\C}_{\kappa}=\C_{\kappa}$. 
\begin{lemma}[Recursive definition of right-inverse]
	Assume the pair $(A,K)$ is in the canonical representation~\eqref{eq:canonical_form} with $U=I$. Let $\tilde{\C}_{k}$ be the upper-left part of the controllability matrix as defined in~\eqref{eq:upper_left_controllability}, with $k\le\kappa$, where $\kappa$ is the controllability index. Let $\Pi_k=H^{-1}_1 A^{\dagger}_{2,1}A^{\dagger}_{3,2}\cdots A^{\dagger}_{k,k-1}$, where $\dagger$ denotes the Moore-Penrose pseudo-inverse. Then, the following inequality holds recursively:
	\begin{equation}\label{eq:pseudo_inverse_recursion}
	  	\snorm{\tilde{\C}^{\dagger}_k}_2\le \snorm{\tilde{\C}^{\dagger}_{k-1}}_2+\snorm{\Pi_{k}}_2+\snorm{\tilde{\C}^{\dagger}_{k-1}\tilde{A}^{k-1}_{k-1}\tilde{H}_{k-1}\Pi_{k}}_2.
	\end{equation}
\end{lemma}
\begin{proof}
The upper-left controllability matrix $\tilde{\C}_k$, $k\le \kappa$ has the following block triangular structure:
	\begin{align}
	\tilde{\C}_k&=\matr{{ccc|c}\tilde{H}_k&\dots&\tilde{A}^{k-1}_k\tilde{H}_k&\tilde{A}^{k-1}_k\tilde{H}_k}\nonumber\\
	&=\matr{{c|c}\tilde{\C}_{k-1}&\tilde{A}^{k-1}_{k-1}\tilde{H}_{k-1}\\0&A_{k,k-1}A_{k-1,k-2}\dots H_1}\label{eq:upper_left_controllability_b}.
	\end{align}
	Based on the above form, we can construct a right-inverse of matrix $\tilde{\C}_k$:
	\begin{align*}
\tilde{\C}^{\sharp}_{k}\triangleq\matr{{cc}\tilde{\C}^{\dagger}_{k-1} & -\tilde{\C}^{\dagger}_{k-1}\tilde{A}^{k-1}_{k-1}\tilde{H}_{k-1}\Pi_{k}\\0& \Pi_{k}},
	\end{align*}
	which satisfies $\tilde{\C}_k\tilde{\C}^{\sharp}_k=I$.
	By the definition of $\tilde{\C}^{\sharp}_k$:
	\[
	\snorm{\tilde{\C}^{\sharp}_k}_2\le \snorm{\tilde{\C}^{\dagger}_{k-1}}_2+\snorm{\Pi_{k}}_2+\snorm{\tilde{\C}^{\dagger}_{k-1}\tilde{A}^{k-1}_{k-1}\tilde{H}_{k-1}\Pi_{k}}_2.
	\]
	To conclude the proof, we invoke Lemma~\ref{lem:right_inverse_lower_bound}.
\end{proof}

	\begin{lemma}\label{lem:right_inverse_lower_bound}
		Let $M\in \R^{s\times t}$ be any matrix with full column rank $s\le t$. Let $M^{\sharp}$ be any right inverse of $M$, i.e. $MM^{\sharp}=I_s$. Then the following inequality is true:
		\[
		\snorm{M^{\dagger}}_2\le \snorm{M^{\sharp}}_2,
		\]
		where $M^{\dagger}$ is the Moore Penrose pseudo-inverse. 
	\end{lemma}
\begin{proof}
	Notice that $M(M^{\dagger}-M^{\sharp})=0$. As a result, we can write $M^{\sharp}=M^{\dagger}+M_{\mathrm{null}}$, where $M_{\mathrm{null}}$ is any matrix in the null space $MM_{\mathrm{null}}=0$. However, the Moorse-Penrose pseudoinverse and $M_{\mathrm{null}}$ are orthogonal
	\[
(M^{\dagger})'M_{\mathrm{null}}=0.
	\]
	By orthogonality, for every $x\in R^{t\times 1}$ we have $\snorm{M^{\sharp}x}_2=\sqrt{\snorm{M^{\dagger}x}^2+\snorm{M_{\mathrm{null}}x}^2}\ge \snorm{M^{\dagger}x}^2$.
\end{proof}
Since all coupling matrices $A_{k,k-1},\dots,A_{2,1},H_1$ have least singular value lower bounded by $\mu$, the product of their pseudo-inverses is upper bounded by:
\[
\snorm{\Pi_k}\le \mu^{-k}.
\] 
So, we should expect~\eqref{eq:pseudo_inverse_recursion} to grow no faster than exponentially with $\kappa$. However, the main challenge is to control the last term in~\eqref{eq:pseudo_inverse_recursion}. Unless we follow a careful analysis, if we just apply the submultiplicative property of the norm we will get bounds which are exponential with $\kappa^2$ instead of $\kappa$. 
The idea is the following. Since by definition $\tilde{C}_{k-1}$ has full rank, then there exists an appropriate matrix $\Lambda_{k-1}\in\R^{(k-1)r\times r_{k}}$ such that
\[
\tilde{A}^{k-1}_{k-1}\tilde{H}_{k-1}\Pi_{k}=\tilde{C}_{k-1}\Lambda_{k-1}.
\]
 Then the above bound becomes:
\begin{equation}\label{eq:pseudo_inverse_recursion_lambda}
	\snorm{\tilde{\C}^{\dagger}_k}\le \snorm{\tilde{\C}^{\dagger}_{k-1}}+\mu^{-k}+\snorm{\Lambda_{k-1}},
\end{equation}
where we used the fact that $\snorm{\tilde{\C}^{\dagger}_{k-1}\tilde{\C}_{k-1}}\le 1$. For the remaining proof, we need to construct such a matrix $\Lambda_{k-1}$ and upper bound it.

\begin{lemma}
Let $\Lambda_{k-2}\in\R^{(k-2)r\times r_{k-1}}$ be any matrix such that:
\[
\tilde{A}^{k-2}_{k-2}\tilde{H}_{k-2}\Pi_{k-1}=\tilde{\C}_{k-2}\Lambda_{k-2}
\]
There exists a matrix $\Lambda_{k-1}\in\R^{(k-1)r\times r_{k}}$ such that:
\[
\tilde{A}^{k-1}_{k-1}\tilde{H}_{k-1}\Pi_{k}=\tilde{\C}_{k-1}\Lambda_{k-1}
\]
with
\begin{equation}\label{eq:pf_upper_bound_auxiliary_matrix}
\norm{\Lambda_{k-1}}_2\le \frac{2+M}{\mu}\norm{\Lambda_{k-2}}_2+\frac{M}{\mu}\snorm{\tilde{\C}^{\dagger}_{k-2}}_2+\mu^{-k}M.
\end{equation}
\end{lemma}
\begin{proof}
\textbf{Part A: algebraic expression for $\tilde{A}^{k-1}_{k-1}\tilde{H}_{k-1}\Pi_{k}$.}
Observe that every matrix $\tilde{A}_{k-1}$ includes the previous as an upper-left submatrix:
\[
\tilde{A}_{k-1}=\matr{{cc}\tilde{A}_{k-2}&A_{1:k-2,k-1}\\
A_{k-1,1:k-2}&A_{k-1,k-1}},
\]
with
\[
A_{1:k-1,k-1}=\matr{{c}A_{1,k-1}\\\vdots\\A_{k-2,k-1}}, \, A_{k-1,1:k-2}=\matr{{cccc}0&\cdots&0&A_{k-1,k-2}}
\]
Let also:
\[
Q_{k}=A_{k,k-1}A_{k-1,k-2}\cdots H_1.
\]
A direct computation gives:
\begin{equation}\label{eq:upper_left_controllability_c}
	\tilde{\C}_{k}=\matr{{ccc}\tilde{\C}_{k-2} &\tilde{A}^{k-2}_{k-2}\tilde{H}_{k-2} &\tilde{A}^{k-1}_{k-2}\tilde{H}_{k-2}+A_{1:k-2,k-1}Q_{k-1} \\0& Q_{k-1}& A_{k-1,1:k-2}\tilde{A}^{k-2}_{k-2}\tilde{H}_{k-2}+A_{k-1,k-1}Q_{k-1}\\
		0&0&Q_{k}}.
\end{equation}
As a result of~\eqref{eq:upper_left_controllability_b} and~\eqref{eq:upper_left_controllability_c},
\[
\tilde{A}^{k-1}_{k-1}\tilde{H}_{k-1}\Pi_{k}=\matr{{c}\tilde{A}^{k-1}_{k-2}\tilde{H}_{k-2}+A_{1:k-2,k-1}Q_{k-1}\\A_{k-1,1:k-2}\tilde{A}^{k-2}_{k-2}\tilde{H}_{k-2}+A_{k-1,k-1}Q_{k-1}}\Pi_{k}.
\]
We can simplify the above expression using $Q_{k-1}\Pi_{k-1}=I$ and $\tilde{A}^{k-2}_{k-2}\tilde{H}_{k-2}\Pi_{k-1}=\tilde{\C}_{k-2}\Lambda_{k-2}$:
\begin{equation}\label{eq:pf_last_column_analysis}
\tilde{A}^{k-1}_{k-1}\tilde{H}_{k-1}\Pi_{k}=\matr{{c}\tilde{A}_{k-2}\tilde{\C}_{k-2}\Lambda_{k-2}+A_{1:k-2,k-1}\\A_{k-1,1:k-2}\tilde{\C}_{k-2}\Lambda_{k-2}+A_{k-1,k-1}}\tilde{A}^{\dagger}_{k,k-1}.
\end{equation}

\noindent\textbf{Part B: last rows as linear combination.}\\
Our goal is to express~\eqref{eq:pf_last_column_analysis} as a linear combination of the columns of:
\[
\tilde{\C}_{k-1}=\matr{{cc}\tilde{\C}_{k-2}&\tilde{A}^{k-2}_{k-2}\tilde{H}_{k-2}\\0&Q_{k-1}}.
\]
Since $\tilde{\C}_{k-1}$ has a triangular structure, we start from the last $r_{k-1}$ rows of $\tilde{A}^{k-1}_{k-1}\tilde{H}_{k-1}\Pi_{k}$
Exploiting the structure of $A_{k-1,1:k-2}$, which includes many zeros we can write:
\begin{align*}
&A_{k-1,1:k-2}\tilde{\C}_{k-2}\Lambda_{k-2}+A_{k-1,k-1}\\
&=\matr{{cccc}0&\cdots&0&A_{k-1,k-2}}\matr{{c|c}\tilde{\C}_{k-3}&\tilde{A}^{k-3}_{k-3}\tilde{H}_{k-1}\\0&A_{k-2,k-3}A_{k-3,k-4}\dots H_1}\Lambda_{k-2}+A_{k-1,k-1}\\
&=A_{k-1,k-2}Q_{k-2}\Lambda_{k-2,k-2}+A_{k-1,k-1}\\
&=Q_{k-1}\Lambda_{k-2,k-2}+A_{k-1,k-1},
\end{align*}
where $\Lambda_{k-2,k-2}\in \R^{r\times r_{k-1}}$ are the last $r$ rows of matrix $\Lambda_{k-2}$: \[\Lambda_{k-2}=\matr{{c}\Lambda_{k-2,1}\\\vdots\\\Lambda_{k-2,k-2}}.\]
Finally, we car rewrite the last $r_{k-1}$ rows of $\tilde{A}^{k-1}_{k-1}\tilde{H}_{k-1}\Pi_{k}$ as:
\begin{align}\label{eq:last_row}
&(A_{k-1,1:k-2}\tilde{\C}_{k-2}\Lambda_{k-2}+A_{k-1,k-1})\tilde{A}^{\dagger}_{k,k-1}=Q_{k-1}(\Lambda_{k-2,k-2}+\Pi_{k-1}A_{k-1,k-1})\tilde{A}^{\dagger}_{k,k-1}
\end{align}

\noindent\textbf{Part c: remaining rows.}\\
From~\eqref{eq:last_row}, we can eliminate the last rows:
\begin{align*}
&\tilde{A}^{k-1}_{k-1}\tilde{H}_{k-1}\Pi_{k}-\matr{{c}\tilde{A}^{k-2}_{k-2}\tilde{H}_{k-2}\\Q_{k-1}}(\Lambda_{k-2,k-2}+\Pi_{k-1}A_{k-1,k-1})\tilde{A}^{\dagger}_{k,k-1}\\
&=\matr{{c}\tilde{A}_{k-2}\tilde{\C}_{k-2}\Lambda_{k-2}+A_{1:k-2,k-1}-\tilde{A}^{k-2}_{k-2}\tilde{H}_{k-2}\Lambda_{k-2,k-2}-\tilde{A}^{k-2}_{k-2}\tilde{H}_{k-2}\Pi_{k-1}A_{k-1,k-1}\\0}\tilde{A}^{\dagger}_{k,k-1}\\
&=\matr{{c}\tilde{A}_{k-2}\tilde{\C}_{k-2}\Lambda_{k-2}+A_{1:k-2,k-1}-\tilde{A}^{k-2}_{k-2}\tilde{H}_{k-2}\Lambda_{k-2,k-2}-\tilde{\C}_{k-2}\Lambda_{k-2}A_{k-1,k-1}\\0}\tilde{A}^{\dagger}_{k,k-1}
\end{align*}
Notice that by the shift structure of the controllability matrix:
\begin{align*}
\tilde{A}_{k-2}\tilde{\C}_{k-2}\Lambda_{k-2}-\tilde{A}^{k-2}_{k-2}\tilde{H}_{k-2}\Lambda_{k-2,k-2}&=\matr{{ccc}\tilde{A}_{k-2}\tilde{H}_{k-2}&\dots &\tilde{A}^{k-2}_{k-2}\tilde{H}_{k-2}}\Lambda_{k-2}-\tilde{A}^{k-2}_{k-2}\tilde{H}_{k-2}\Lambda_{k-2,k-2}\\
&=\matr{{cccc}\tilde{A}_{k-2}\tilde{H}_{k-2}&\dots &\tilde{A}^{k-3}_{k-2}\tilde{H}_{k-2}&0}\Lambda_{k-2}\\
&=\matr{{cccc}\tilde{H}_{k-2}&\tilde{A}_{k-2}\tilde{H}_{k-2}&\dots &\tilde{A}^{k-3}_{k-2}\tilde{H}_{k-2}}\Lambda^{\mathrm{shift}}_{k-2}\\
&=\tilde{\C}_{k-2}\Lambda^{\mathrm{shift}}_{k-2}.
\end{align*}
where
\[
\Lambda^{\mathrm{shift}}_{k-2}=\matr{{c}0\\\Lambda_{k-2,1}\\\vdots\\\Lambda_{k-2,k-3}}.
\]
Moreover, we can write $A_{1:k-2,k-1}=\tilde{\C}_{k-2}\tilde{\C}^{\dagger}_{k-2}A_{1:k-2,k-1}$

\noindent\textbf{Part d: construction of $\Lambda_{k-1}$.}\\
Combining the above equalities:
\begin{align*}
	&\tilde{A}^{k-1}_{k-1}\tilde{H}_{k-1}\Pi_{k}=\matr{{c}\tilde{A}^{k-2}_{k-2}\tilde{H}_{k-2}\\Q_{k-1}}(\Lambda_{k-2,k-2}+\Pi_{k-1}A_{k-1,k-1})\tilde{A}^{\dagger}_{k,k-1}\\
	&+\matr{{c}\tilde{\C}_{k-2}\\0}(\Lambda^{\mathrm{shift}}_{k-2}+\tilde{\C}^{\dagger}_{k-2}A_{1:k-2,k-1}-\Lambda_{k-2}A_{k-1,k-1})\tilde{A}^{\dagger}_{k,k-1}.
\end{align*}
Hence we can select:
\[
\Lambda_{k-1}=\matr{{c}\paren{\Lambda^{\mathrm{shift}}_{k-2}+\tilde{\C}^{\dagger}_{k-2}A_{1:k-2,k-1}-\Lambda_{k-2}A_{k-1,k-1}}\tilde{A}^{\dagger}_{k,k-1}\\\paren{\Lambda_{k-2,k-2}+\Pi_{k-1}A_{k-1,k-1}}\tilde{A}^{\dagger}_{k,k-1}},
\]
with
\[
\norm{\Lambda_{k-1}}\le (2+M)\mu^{-1}\norm{\Lambda_{k-2}}+M\mu^{-1}\snorm{\tilde{\C}^{\dagger}_{k-2}}+\mu^{-k}M
\]
\end{proof}

Now we can complete the proof of Theorem~\ref{thm:upper_bound_control_energy}. It is sufficient to select $\Lambda_1$: 
\[
A_{1,1}H_1\Pi_2=H_1H^{-1}_1A_{1,1}A^{\dagger}_{2,1}=\tilde{\C}_1\Lambda_1,
\]
with $\norm{\Lambda_1}_2\le M\mu^{-2}$.
Let $\alpha_k=\matr{{ccc}\snorm{\tilde{\C}^{\dagger}_k}&\snorm{\Lambda_k}&\mu^{-k}}'$. From~\eqref{eq:pseudo_inverse_recursion_lambda},~\eqref{eq:pf_upper_bound_auxiliary_matrix} we obtain the following recursion:
\[
\alpha_k\le \matr{{ccc}1&1&\mu^{-1}\\\frac{M}{\mu}&\frac{2+M}{\mu}&\frac{M}{\mu}\\0&0&\mu^{-1}}\alpha_{k-1},
\]
where the inequality is interpreted coordinate-wise. Let $\Xi$ be the matrix of the above recursion.
We have the crude bound:
\[
\snorm{\C^{\dagger}_\kappa}_2=\snorm{\tilde{\C}^{\dagger}_\kappa}_2\le \snorm{\Xi^{\kappa-1}}_2\snorm{\alpha_1}_{2},
\]
where $\snorm{\Xi^{\kappa}}_2\snorm{\alpha_1}_{2}\le\poly^{\kappa}(M/\mu)$. This completes the proof.

\end{document}